\documentclass[11pt]{amsart}
\usepackage{amssymb,mathrsfs,graphicx,enumerate}
\usepackage{amsmath,amsfonts,amssymb,amscd,amsthm,bbm}
\usepackage{kotex}
\allowdisplaybreaks
\usepackage[retainorgcmds]{IEEEtrantools}
\usepackage{colortbl}
\usepackage{graphicx, subfigure}

\DeclareMathOperator*{\argmin}{arg\,min} 
\DeclareMathOperator*{\argmax}{arg\,max} 

\title[The Lohe Hermitian sphere model with adapted couplings]{Asymptotic interplay of states and adapted coupling gains in the Lohe hermitian sphere model}

\author[Byeon]{Junhyeok Byeon}
\address[Junhyeok Byeon]{\newline Department of Mathematical Sciences\newline Seoul National University, Seoul 08826, Republic of Korea}
\email{giugi2486@snu.ac.kr}

\author[Ha]{Seung-Yeal Ha}
\address[Seung-Yeal Ha]{\newline Department of Mathematical Sciences and Research Institute of Mathematics \newline Seoul National University, Seoul 08826 and \newline
Korea Institute for Advanced Study, Hoegiro 85, Seoul 02455, Republic of Korea}
\email{syha@snu.ac.kr}

\author[Park]{Hansol Park}
\address[Hansol Park]{\newline Department of Mathematical Sciences\newline Seoul National University, Seoul 08826, Republic of Korea}
\email{hansol960612@snu.ac.kr}

\newtheorem{theorem}{Theorem}[section]
\newtheorem{lemma}{Lemma}[section]

\newtheorem{proposition}{Proposition}[section]
\newtheorem{remark}{Remark}[section]

\newtheorem{definition}{Definition}[section]

\newcommand{\bbr}{\mathbb R}

\newcommand{\bbh}{\mathbb H}
\newcommand{\bbs}{\mathbb S}

\newcommand{\bbc}{\mathbb C}

\begin{document}

\date{\today}


\thanks{\textbf{Acknowledgment.} The work of S.-Y. Ha was supported by National Research Foundation of Korea(NRF-2020R1A2C3A01003881).
The work of H. Park was supported by Basic Science Research Program through the National Research Foundation of Korea(NRF) funded by the Ministry of Education (2019R1I1A1A01059585) }

\begin{abstract}
We study emergent dynamics of the Lohe hermitian sphere (LHS) model with the same free flows under the dynamic interplay between state evolution and adaptive couplings. The LHS model is a complex counterpart of the Lohe sphere (LS) model on the unit sphere in Euclidean space, and when particles lie in the Euclidean unit sphere embedded in $\bbc^{d+1}$, it reduces to the Lohe sphere model. In the absence of interactions between states and coupling gains, emergent dynamics have been addressed in \cite{H-P2}. In this paper, we further extend earlier results in the aforementioned work to the setting in which the state and coupling gains are dynamically interrelated via two types of coupling laws, namely anti-Hebbian and Hebbian coupling laws. In each case, we present two sufficient frameworks leading to complete aggregation depending on the coupling laws, when the corresponding free flow is the same for all particles.
\end{abstract}

\maketitle \centerline{\date}


\section{Introduction}
Collective behaviors of classical and quantum systems are ubiquitous, e.g., aggregation of bacteria, schooling of fishes, flocking of birds and synchronous firing of fireflies and neurons, etc \cite{A-B, A-B-F, B-H, B-B, Ku2, Pe, P-R, St, T-B-L, T-B, VZ, Wi2, Wi1}. These coherent phenomena were first modeled by two pioneers, Arthur Winfree \cite{Wi2} and Yoshiki Kuramoto \cite{Ku2} in almost half-century ago, and after their pioneering works, several mathematical models were proposed and studied from the viewpoint of collective behaviors. Among them, our main interest in this paper lies in the LHS model \cite{H-P2} which  corresponds to the special case of the Lohe tensor model \cite{H-P1}. The Lohe tensor model is a natural higher-dimensional extension of low-dimensional aggregation models such as the Kuramoto model \cite{A-B, B-C-M, C-H-J-K, C-S, D-X, D-B1, D-B, H-K-R},  sphere models \cite{C-H5, HKLN,  J-C, Lo-1, Lo-2, M-T-G, T-M, Zhu} and matrix models \cite{B-C-S,  D-F-M-T, De, Kim, Lo-0}. Before we move onto the description of the LHS model, we first set the hermitian unit sphere $\bbh\bbs^d$ which is the unit sphere in $\bbc^{d+1}$ centered at the origin:
\begin{align*}
\begin{aligned}
& z = ([z]_1, \cdots, [z]_{d+1}) \in \bbc^{d+1}, \quad {\tilde z} = ([\tilde z]_1, \cdots,[\tilde z]_{d+1}) \in \bbc^{d+1}, \\
& \langle z, {\tilde z} \rangle := \sum_{\alpha = 1}^{d+1} \overline{[z]_{\alpha}} [\tilde z]_{\alpha}, \quad \|z \| := \sqrt{ \langle z, z \rangle}, \quad  \mathbb{HS}^d := \{ z \in \mathbb{C}^{d+1} \, | \, \|z\| = 1\},
\end{aligned}
\end{align*}
where $\overline{[z]_{\alpha}}$ is the complex conjugate of $[z]_{\alpha}$.  Equipped with these notation, the LHS model with the same free flow reads as follows:
\begin{equation}  \label{A-0}
\displaystyle  \dot{z}_j=\Omega z_j+\displaystyle\frac{1}{N}\sum_{k=1}^N \kappa_{jk}\big(\langle z_j, z_j\rangle z_k-\langle z_k, z_j\rangle z_j\big) +\frac{1}{N}\sum_{k=1}^N {\lambda}_{jk}\big(\langle z_j, z_k\rangle-\langle z_k,z_j\rangle\big)z_j,
\end{equation}
where $\kappa_{jk}$ and $\lambda_{jk}$ are constant coupling gains coined as  ``{\it Lohe sphere coupling gain}" and ``{\it rotational coupling gain}" respectively.  Here $\Omega$ is a $(d+1) \times (d+1)$ skew-hermitian matrix:
\[ \Omega^\dagger = -\Omega, \quad j \in {\mathcal N} := \{ 1, \cdots, N \}, \]
where $\Omega^\dagger$ is the Hermitian conjugate of $\Omega$. \newline

In this paper, we are interested in the following simple question: \newline

\begin{quote}
``What if dynamics of coupling gains interacts with the dynamics of states? i.e., dynamic interplay between coupling gains and state evolution. In this case, under what conditions, can a coupled system exhibit emergent dynamics?"
\end{quote}

\vspace{0.2cm}

The above question has been addressed in other aggregation models, e.g., the Kuramoto model with adaptive couplings \cite{HLLP, HNP}, the Lohe sphere model with adaptive couplings \cite{HKLN}, the Lohe matrix model with adaptive couplings \cite{Kim}. Then the coupled dynamics for $\{ (z_j, \kappa_{jk}, \lambda_{jk})$ is governed by the Cauchy problem to the LHS-AC model:
\begin{align} 
\begin{aligned} \label{A-1}
& \dot{z}_j =\Omega z_j +  \frac{1}{N}\sum_{k=1}^N\kappa_{jk}\Big(\langle z_j, z_j\rangle z_k-\langle z_k, z_j\rangle z_j\Big)
+\frac{1}{N}\sum_{k=1}^N\lambda_{jk}\Big(\langle z_j, z_k\rangle-\langle z_k,z_j\rangle \Big)z_j,  \vspace{0.2cm}\\
&  \dot{\kappa}_{jk} = -\gamma_0\kappa_{jk}+\mu_0\Gamma_0(z_j, z_k),\quad
\dot{\lambda}_{jk}= -\gamma_1\lambda_{jk}+\mu_1\Gamma_1(z_j, z_k), \quad t > 0, \vspace{0.2cm}\\
& (z_j, \kappa_{jk}, \lambda_{jk})(0) = (z_j^0, \kappa_{jk}^0, \lambda_{jk}^0) \in \bbh\bbs^{d} \times \bbr_+ \times \bbr, \quad  j, k \in \mathcal{N},
\end{aligned}
\end{align}
where $\mu_k$ and $\gamma_k$ are positive constants. \newline

Throughout the paper, we use the following handy notation:
\[ Z := (z_1, \cdots, z_N)\in(\mathbb{HS}^{d})^N, \quad K := [\kappa_{ij}], \quad \Lambda := [\lambda_{ij}], \quad  \max_{i,j} := \max_{1 \leq i, j \leq N}, \quad \min_{i,j} := \min_{1 \leq i, j \leq N}. \]
Before we discuss our main results, we recall the concept of ``{\it complete aggregation}" for the Cauchy problem \eqref{A-1} as follows.
\begin{definition} \label{D1.1}
Let $(Z, K, \Lambda)$ be a solution to \eqref{A-1}. Then, complete aggregation occurs asymptotically if and only if following relations hold.
\begin{equation*} \label{A-6}
\lim_{t \to \infty} \max_{1 \leq i,j \leq N} \|z_i(t) - z_j(t) \| = 0.
\end{equation*}
\end{definition}
Recall that the primary purpose of this paper is to provide sufficient frameworks leading to complete aggregation for system \eqref{A-1}. In general, there will be no functional dependence between $\kappa_{jk}$ and $\lambda_{jk}$. From now on, we assume that the system parameters satisfy the following relations:
\begin{equation} \label{A-1-1}
\Omega = 0, \quad \gamma_0 = \gamma_1 = \gamma, \quad \mu_0 = \mu_1  = \mu.
\end{equation}
Motivated by the reduction from the Stuart-Landau(SL) model to the LHS model in Section \ref{sec:2.3}, we call the following relation as the SL coupling gain pair:
\begin{equation}\label{A-2}
\kappa_{jk} > 0, \quad \lambda_{jk}=-\frac{1}{2} \kappa_{jk}, \quad j, k \in {\mathcal N}.
\end{equation}
Under the setting \eqref{A-2}, due to Lemma \ref{L2.2} and Lemma \ref{L2.3}, system \eqref{A-1} becomes
\begin{equation} \label{A-3}
\begin{cases}
\displaystyle \dot{z}_j = \frac{1}{N}\sum_{k=1}^N \kappa_{jk} \left[ z_k-\frac{1}{2}\Big(\langle z_j, z_k \rangle+\langle z_k, z_j \rangle \Big)z_j\right ], \quad t > 0, \vspace{0.2cm} \\
\displaystyle \dot{\kappa}_{jk} = -\gamma \kappa_{jk} + \mu \Gamma_0(z_j, z_k), \vspace{0.2cm} \\
\displaystyle (z_j, \kappa_{ij})(0) =(z_j^0, \kappa_{ij}^0) \in \bbh\bbs^{d} \times \bbr_+, \quad i, j\in \mathcal{N}.
\end{cases}
\end{equation}
At the end of Section \ref{sec:2}, we will see that system \eqref{A-3} on $\bbc^{(d+1)N}$ can be rewritten as the Lohe sphere model on $\bbr^{2(d+1)N}$ for a special case. Now, we set  
\[ \tilde{\lambda}_{jk} :=\frac{1}{2}\kappa_{jk}+\lambda_{jk}, \quad j, k \in {\mathcal N}, \quad \tilde{\Gamma}(z, w) :=\frac{1}{2}\Gamma_0(z, w)+\Gamma_1(z, w), \quad z, w \in \bbh\bbs^{d}.  \]
Then, under the setting \eqref{A-1-1}, system \eqref{A-1} can be rewritten as a perturbed system of \eqref{A-3}:
\begin{equation}\label{A-4}
\begin{cases}
\displaystyle \dot{z}_j= \frac{1}{N}\sum_{k=1}^N\kappa_{jk} \left[ z_k-\frac{1}{2}\Big(\langle z_j, z_k \rangle+\langle z_k, z_j \rangle \Big)z_j\right]+\frac{1}{N}\sum_{k=1}^N\tilde{\lambda}_{jk}\left(\langle z_j, z_k\rangle-\langle z_k, z_j\rangle\right)z_j,  \vspace{0.2cm} \\
\displaystyle \dot{\kappa}_{jk}=-\gamma\kappa_{jk}+\mu\Gamma_0(z_j,z_k),\quad \dot{\tilde{\lambda}}_{jk}=-\gamma\tilde{\lambda}_{jk}+\mu\tilde{\Gamma}(z_j, z_k),\quad t > 0, \vspace{0.2cm} \\
\displaystyle (z_j, \kappa_{jk}, \lambda_{jk})(0) =(z_j^0, \kappa_{jk}^0, \lambda_{jk}^0) \in \bbh\bbs^{d} \times \bbr_+ \times \bbr, \quad j, k \in {\mathcal N}.
\end{cases}
\end{equation}
and we take the following ansatz for the coupling law $\Gamma_0$:
\begin{equation} \label{A-5}
\Gamma_0(w, z) = 
\begin{cases}
\displaystyle \|w-z\|^2 \quad &: \mbox{Anti-Hebbian coupling law}, \\
\displaystyle 1-\frac{1}{2}\|w-z\|^2 \quad &: \mbox{Hebbian coupling law}. 
\end{cases}
\end{equation}
The choice and meaning of Anti-Hebbian and Hebbian coupling laws will be elaborated in Section \ref{sec:3}. When the coupling gains $\kappa_{jk}$ and $\lambda_{jk}$ are simply positive constants and uniformly independent of $j$ and $k$, emergent dynamics of \eqref{A-1} has been extensively studied in \cite{H-P2}. However, for the coupled system \eqref{A-2}, we will see that coupling gains tend to zero asymptotically. Hence, our presented results do not overlap with the results in aforementioned work. As complete aggregation occurs asymptotically, the vanishing of coupling gains is natural in some sense, because the coupling gain will not be needed, once complete aggregation is achieved. \newline

In what follows, we briefly discuss main results of this paper. First, we study emergent behaviors of \eqref{A-3} under \eqref{A-5} for $\Gamma_0$.  For the anti-Hebbian coupling law, we use the following Lyapunov functional measuring the degree of aggregation:
\begin{equation} \label{A-7}
\mathcal{L}_{ij}=\frac{1}{2}\|z_i-z_j\|^2+\frac{1}{4\mu N}\sum_{k=1}^{N}(\kappa_{ik} -\kappa_{jk})^2.
\end{equation}
Our first result deals with  \eqref{A-3} with anti-Hebbian coupling law $\eqref{A-5}_1$. When initial data satisfy following relation:
\[
 \max_{i,j} \mathcal{L}_{ij}^0 <1, 
 \]
complete aggregation emerges and mutual coupling gains tend to zero asymptotically (see Theorem \ref{T3.1}):
\[
\lim_{t \to \infty} \| z_i(t)-z_j(t) \| = 0 \quad \text{and} \quad \lim_{t \to \infty} \kappa_{ij}(t) = 0, \quad i, j \in {\mathcal N}.
\]
Our second result is conserved with  \eqref{A-3} incorporated by Hebbian coupling function $\eqref{A-5}_2$. In this case, instead of \eqref{A-7}, we introduce another functional:
\begin{equation*} \label{A-8}
{\mathcal D}(Z) := \frac{1}{2} \max_{i,j} \|z_i-z_j\|^2,
\end{equation*}
which is the half of square of state diameter. \newline

If there exist a constant $\kappa$ satisfying the following relations
\begin{equation*}
0 < \kappa < \min \left\{  \frac{\mu}{\gamma},~\min_{i,j}\kappa_{ij}^0 \right \}, \quad \max \left\{  \max_{i,j} \kappa^0_{ij}, ~\frac{\mu}{\gamma} \right \} \leq \frac{2\mu \kappa}{2\mu -\gamma \kappa}, 
\quad {\mathcal D}(Z^0) < 1-\frac{\gamma}{\mu} \kappa,
\end{equation*}
then, there exist positive constants $C_0>0$ and $C_1>0$ satisfying
\[
{\mathcal D}(Z(t)) \leq C_0 e^{-C_1 t}, \quad t>0.
\]
We refer to Theorem \ref{T3.2} and Section \ref{sec:4}  for details. \newline

Secondly, we study emergent behaviors of \eqref{A-4} with a general coupling gain pair:
\[ \kappa_{ij} > 0, \quad \lambda_{ij} \in \bbr, \quad\forall~ i, j \in {\mathcal N}. \]
Our third result can be stated as follows. Suppose that system parameters and initial data satisfy 
\[
\begin{cases}
\displaystyle \tilde{\lambda}_{ij}^0={\tilde \lambda}^0, \quad  i,j \in \mathcal{N}, \quad \tilde{\Gamma}(t) = 0, \quad \forall t>0, \quad \vspace{0.2cm} \\
\displaystyle
\max_{i,j}\frac{2|\tilde\lambda^0|}{\kappa_{ij}^0}+\max_{k,l}\mathcal{L}_{kl}^0 < 1,
\end{cases}
\]
where $\mathcal{L}_{ij}^0 := {\mathcal L}_{ij}(Z^0, K^0)$ and $\tilde{\lambda}_{ij}:=\frac{1}{2}\kappa_{ij}+\lambda_{ij}$. \newline

Then under anti-Hebbian coupling law $\eqref{A-5}_1$, we have complete aggregation and vanishing of coupling gains (see Theorem \ref{T3.3}):
\[
\lim_{t \to \infty} \| z_i(t)-z_j(t) \| = 0 \quad \text{and} \quad \lim_{t \to \infty} \kappa_{ij}(t) = 0.
\]
Finally, suppose there exist a constant $\kappa$ such that
\begin{equation*}
2|\tilde{\lambda}^0| < \kappa < \min \left\{  \frac{\mu}{\gamma},~\min_{i,j}\kappa_{ij}^0 \right \}, \quad 
\max \left\{  \max_{i,j} \kappa^0_{ij}, ~\frac{\mu}{\gamma} \right \} \leq \frac{2\mu (\kappa-2|\tilde{\lambda^0}|)}{2\mu-\gamma \kappa}, 
\quad {\mathcal D}(Z^0) < 1-\frac{\gamma}{\mu} \kappa,
\end{equation*}
and let $(Z,K)$ be a solution to \eqref{A-4}. Then under Hebbian coupling law $\eqref{A-5}_2$, there exist positive constants $C_2>0$ and $C_3>0$ satisfying
\[
{\mathcal D}(Z(t)) \leq C_2 e^{-C_3 t}, \quad t>0.
\]
See Theorem \ref{T3.4} for details. 
%

\vspace{0.5cm}

The rest of the paper is organized as follows.
In Section \ref{sec:2}, we present basic properties of the LHS-AC model \eqref{A-1}, its relations with  previous aggregation models and a reduction from the generalized Stuart-Landau model to the LHS model with special coupling pair  \eqref{A-2}. In Section \ref{sec:3}, we briefly summarize our main results on the emergent collective behaviors of \eqref{A-1}. In Section \ref{sec:4}, we study the emergent dynamics of \eqref{A-3}. In Section \ref{sec:5}, we study emergent dynamics of \eqref{A-4}. Finally, Section \ref{sec:6} is devoted to a brief summary of our main results and some remaining issues to be addressed in a future work.

\section{Preliminaries} \label{sec:2}
\setcounter{equation}{0}
In this section, we study several basic properties of the LHS-AC model  \eqref{A-0} and its relations with other first-order aggregation models with emergent property. We also provide a reduction from a generalized Stuart-Landau model to the LHS model.
\subsection{Basic estimates} \label{sec:2.1}
In this subsection, we study basic properties of system \eqref{A-1} such as the positivities of coupling gains, conservation of modulus of $z_i$ and solution splitting property.
\begin{lemma}\label{L2.1}
\emph{(Positivity and symmetry of coupling gains)}
Suppose that the coupling laws $\Gamma_0$ and $\Gamma_1$ take nonnegative values:
\[ \Gamma_0(z, \tilde z) \geq 0, \quad  \Gamma_1(z, \tilde z) \geq 0, \quad z, {\tilde z} \in \bbh \bbs^{d}, \]
and let $(Z, K, \Lambda)$ be a solution to \eqref{A-1}. Then, we have the following assertions:
\begin{enumerate}
\item 
If initial coupling gains satisfy
\[
\kappa_{ij}^0>0, \quad \lambda_{ij}^0>0, \quad  \forall~i,j \in \mathcal{N}, 
\]
then one has positivities of coupling gains:
\[
\kappa_{ij}(t)>0, \quad \lambda_{ij}(t) > 0, \quad\forall~ t>0, \quad   \in i,j \mathcal{N}.
\]
\item 
If initial coupling gains satisfy
\[
\kappa_{ij}^0=\kappa_{ji}^0, \quad \lambda_{ij}^0 = \lambda_{ji}^0, \quad\forall~  i,j \in \mathcal{N}, 
\]
then symmetries of the coupling gains are preserved:
\[
\kappa_{ij}(t)=\kappa_{ji}(t), \quad \lambda_{ij}(t) = \lambda_{ji}(t), \quad\forall~ t > 0, \quad  i,j \in \mathcal{N}.
\]
\end{enumerate}
\end{lemma}
\begin{proof} (i)~For the first assertion, we use $\eqref{A-1}_2$ and Duhamel's principle to find the following representations: for $t \geq 0$, 
\begin{align}
\begin{aligned} \label{N-B-1}
\kappa_{ij}(t) &= e^{-\gamma_0 t}\left(\kappa_{ij}^0 + \int_0^t \mu_0 e^{\gamma_0 s} \Gamma_0(z_i(s),z_j(s))ds \right), \\
\lambda_{ij}(t) &= e^{-\gamma_1 t}\left(\lambda_{ij}^0 + \int_0^t \mu_1 e^{\gamma_1 s} \Gamma_1(z_i(s),z_j(s))ds \right).
\end{aligned}
\end{align}
Since system parameters $\mu_k$ and $\gamma_k$ are nonnegative, it follows from \eqref{N-B-1} that 
\[ \kappa_{ij}(t) > 0, \quad \lambda_{ij}(t) > 0, \quad t > 0. \]
(ii)~For the second assertion, we use the symmetry of $\kappa_{ij}^0,~\lambda_{ij}^0,~\Gamma_0$ and $\Gamma_1$ in the index exchange $i \longleftrightarrow j$ to find the desired results. 
\end{proof}
\begin{lemma}\label{L2.2}
\emph{(Conservation of modulus)}
Let $(Z, K, \Lambda)$ be a solution to the Cauchy problem
 \eqref{A-1}. Then, the modulus $\| z_i \|$ is a conserved quantity:~for $i \in {\mathcal N}$, 
\[ \|z_i^0\|=1 \quad \Longrightarrow \quad \|z_i(t) \| = 1, \quad t > 0. \]
\end{lemma}
\begin{proof}
It follows from the symmetry of coupling strengths that 
\begin{align}
\begin{aligned} \label{B-0}
\langle z_i,\dot{z_i} \rangle &= \langle z_i, \Omega z_i \rangle +\displaystyle\frac{1}{N}\sum_{k=1}^N\kappa_{ik} \Big (\langle z_i, z_k \rangle -\langle z_k, z_i\rangle \Big ) \langle z_i, z_i \rangle +\frac{1}{N}\sum_{k=1}^N\lambda_{ik} \Big(\langle z_i, z_k\rangle-\langle z_k,z_i\rangle \Big) \langle z_i, z_i \rangle,\\
\langle \dot{z}_i, z_i \rangle &= \langle \Omega z_i, z_i \rangle +\displaystyle\frac{1}{N}\sum_{k=1}^N\kappa_{ik} \Big( \langle z_k, z_i \rangle -\langle z_i, z_k \rangle \Big) \langle z_i, z_i \rangle +\frac{1}{N}\sum_{k=1}^N \lambda_{ik} \Big(\langle z_k, z_i\rangle-\langle z_i,z_k\rangle \Big)\langle z_i, z_i \rangle.
\end{aligned}
\end{align}
Since $\Omega_i$ is skew-hermitian, we have 
\begin{equation} \label{B-0-0}
\langle z_i, \Omega z_i \rangle + \langle \Omega z_i, z_i \rangle = 0. 
\end{equation}
Finally, we combine \eqref{B-0} and \eqref{B-0-0} to obtain the desired estimate:
\[
\frac{d}{dt}\|z_i\|^2=\frac{d}{dt}\langle z_i,z_i \rangle = \langle z_i, \dot{z}_i \rangle + \langle \dot{z}_i, z_i \rangle = 0.
\]
\end{proof}
Now, we consider corresponding linear and nonlinear flows:
\begin{equation}  \label{B-1-0}
\begin{cases}
\displaystyle \dot{f}_j = \Omega f_j, \quad t > 0,~~\forall~j \in {\mathcal N}, \\
\displaystyle f_j(0) = f_j^0,
\end{cases}
\end{equation}
and 
\begin{equation}\label{B-1-1}
\begin{cases}
\displaystyle \dot{w}_j= \frac{1}{N}\sum_{k=1}^N\kappa_{jk}\big(w_k-\langle w_k, w_j\rangle w_j\big)+\frac{1}{N}\sum_{k=1}^N\lambda_{jk}\big(\langle w_j, w_k\rangle-\langle w_k,w_j\rangle\big)w_j,\vspace{0.1cm}\\
\displaystyle \dot{\kappa}_{jk}= -\gamma_0\kappa_{jk}+\mu_0\Gamma_0(w_j, w_k),\quad
\dot{\lambda}_{jk}= -\gamma_1\lambda_{jk}+\mu_1\Gamma_1(w_j, w_k),\vspace{0.1cm}\\
(w_j(0), \kappa_{jk}(0), \lambda_{jk}(0)) = (w_j^0, \kappa_{jk}^0, \lambda_{jk}^0) \in \bbh\bbs^d \times \bbr_+ \times \bbr, \quad i, j\in {\mathcal N}.
\end{cases} 
\end{equation}

\vspace{0.2cm}

Let $R$ and $L_j$ be solution operators to \eqref{B-1-0} and $\eqref{B-1-1}$, respectively. Then, solutions to \eqref{B-1-0} and $w_j$ in $\eqref{B-1-1}$ can be represented as follows.
\[
f_j(t) = R(t) f_j^0  =: e^{\Omega t} f_j^0, \quad  w_j(t) := L_j(t) (W^0,K^0,\Lambda^0), \quad \forall~j \in  {\mathcal N}.
\]
In next lemma, we show that the full solution operator to \eqref{A-1} can be expressed as a composition of $R$ and $L_j$.
\begin{lemma} \label{L2.3}
\emph{(Solution splitting property)} Let $(Z, K, \Lambda)$ be a solution to system \eqref{A-1} with initial data $(Z^0, K^0, \Lambda^0)$ satisfying 
\[ \Omega_j \equiv \Omega, \quad j \in \mathcal{N}. \]
Then, $z_j$ can be decomposed as a composition of $f_j$ and $w_j$:
\[  z_j(t) = R(t) \circ L_j(t)(Z^0,K^0,\Lambda^0). \]
\end{lemma}
\begin{proof} We substitute $z_j=e^{\Omega t}w_j$ into \eqref{A-0} to obtain
\begin{align}
\begin{aligned} \label{B-1-2}
e^{\Omega t}\dot{w}_j + \Omega e^{\Omega t}w_j =\Omega e^{\Omega t}w_j+\displaystyle\frac{1}{N}\sum_{k=1}^N\kappa_{jk}\big(e^{\Omega t}w_k-\langle e^{\Omega t}w_k, e^{\Omega t}w_j\rangle e^{\Omega t}w_j\big)\\
+\frac{1}{N}\sum_{k=1}^N\lambda_{jk}\big(\langle e^{\Omega t}w_j, e^{\Omega t}w_k\rangle-\langle e^{\Omega t}w_k,e^{\Omega t}w_j\rangle\big)e^{\Omega t}w_j.
\end{aligned}
\end{align}
On the other hand, we use the skew-hermitian property of  $\Omega$ to find
\begin{equation} \label{B-1-3}
    (e^{\Omega t})^{\dagger} = (e^{\Omega t})^{-1}.
\end{equation}
Finally, we combine \eqref{B-1-2} and \eqref{B-1-3} to get 
\[
\dot{w}_j  = \displaystyle\frac{1}{N}\sum_{k=1}^N\kappa_{jk}\big(w_k-\langle w_k, w_j\rangle w_j\big) +\frac{1}{N}\sum_{k=1}^N\lambda_{jk}\big(\langle w_j, w_k\rangle-\langle w_k,w_j\rangle\big)w_j,
\]
so that $w_j(t)=L_j(t)(W^0, K^0, \Lambda^0)$. This gives a desired result.
\end{proof}

\vspace{0.5cm}

Note that the LHS model $\eqref{A-1}_1$ contains two terms involving with $\kappa_{jk}$ and $\lambda_{jk}$.  To see the role of each coupling gain separately, we consider the following subsystems:  \newline

\noindent$\bullet$ (Subsystem A): If we impose the following conditions on \eqref{A-1}:
\begin{align*}
\lambda^0_{jk}=0\quad \text{for all} \quad j, k\in \mathcal{N}, \quad \text{and} \quad\mu_1=0,
\end{align*}
then we have 
\[ \lambda_{jk}(t)=0, \quad \forall~t > 0, \quad  j, k\in \mathcal{N}. \] 
In this case, we have Subsystem A:
\begin{align}\label{B-2}
\begin{cases}
\dot{z}_j=\Omega z_j+\displaystyle\frac{1}{N}\sum_{k=1}^N\kappa_{jk}\big(\langle z_j, z_j\rangle z_k-\langle z_k, z_j\rangle z_j\big), \quad t > 0, \vspace{0.2cm}\\
\dot{\kappa}_{jk}= -\gamma_0\kappa_{jk}+\mu_0\Gamma_0(z_j, z_k), \quad  j, k \in \mathcal{N}, \vspace{0.2cm}\\
(z_j(0), \kappa_{jk}(0)) = (z_j^0, \kappa_{jk}^0) \in \bbh\bbs^{d} \times \bbr_+.
\end{cases}
\end{align}

\vspace{0.5cm}

\noindent$\bullet$ (Subsystem B): If we impose the following condition on \eqref{A-1}:
\begin{align*}
\kappa_{jk}^0=0\quad \text{for all} \quad j, k\in \mathcal{N}, \quad \text{and} \quad\mu_0=0,
\end{align*}
then we have 
\[ \kappa_{jk}(t)=0 \quad \forall~ t > 0, \quad j, k\in \mathcal{N}. \] 
In this case, one has Subsystem B:
\begin{align}\label{B-3}
\begin{cases}
\dot{z}_j=\Omega z_j+\displaystyle\frac{1}{N}\sum_{k=1}^N\lambda_{jk}\big(\langle z_j, z_k\rangle-\langle z_k,z_j\rangle\big)z_j, \quad t > 0, \vspace{0.2cm}\\
\dot{\lambda}_{jk}= -\gamma_1 \lambda_{jk}+\mu_1\Gamma_1(z_j, z_k), \quad  i, j\in \mathcal{N}, \vspace{0.2cm}\\
(z_j(0), \lambda_{jk}(0)) = (z_j^0, \lambda_{jk}^0) \in \bbh\bbs^d \times \bbr_+.
\end{cases}
\end{align}

\subsection{Reductions to other aggregation models}  \label{sec:2.2}
In this subsection, we study the relations between \eqref{A-1} and other first-order aggregation models with adaptive couplings. \newline

For a real vector-valued state $\{z_1, z_2, \cdots, z_N\} \subset \bbr^{d+1}$, the interaction terms in Subsystem B become zero:
\[  \langle z_j, z_k\rangle-\langle z_k,z_j\rangle = 0, \quad j, k \in {\mathcal N}\]
so that Subsystem B reduces to the free flow:
 \begin{equation*}
 \begin{cases} \label{B-3-0}
 \displaystyle \dot{z}_j =\Omega z_j,  \quad t > 0, \quad  j, k \in {\mathcal N}, \\
 \displaystyle \dot{\lambda}_{jk} = -\gamma_1 \lambda_{jk} +\mu_1\Gamma_1(z_j, z_k).
\end{cases}
\end{equation*}

In next lemma, we show that the real-valuedness of the components of $z_j$ and $\kappa_{ij}$ are propagated along system \eqref{B-2}.
\begin{lemma} \label{L2.4}
Let $(Z, K, \Lambda)$ be a solution to \eqref{B-2} with initial data satisfying following conditions:
\begin{align*}\label{B-3-1}
z_j^0  \in\bbr^{d+1},\quad \Omega \in\bbr^{(d+1)\times (d+1)}, \quad \kappa_{jk}=\kappa_{kj}, \quad \forall~j, k \in {\mathcal N}.
\end{align*}
Then, one has 
\[
z_j(t)\in\bbr^{d+1}, \quad t \geq 0, \quad  j \in {\mathcal N}.
\] 
\end{lemma}
\begin{proof}
Since  the R.H.S. of system $\eqref{B-2}_1$ is Lipschitz continuous with respect to state variables and uniformly bounded, global well-posedness of classical solutions are guaranteed by the standard Cauchy-Lipschitz theory. Meanwhile, governing system $\eqref{B-2}_1$ coincides with the LS model on the sphere in $\bbr^{d+1}$. On the other hand, the LS model has a unique solution which is bounded in $\mathbb{R}^{d+1}$. Thus, we have a desired result.
\end{proof}

By Lemma \ref{L2.4}, if we assume that $(Z,K, \Lambda)$ is a solution to \eqref{B-2} with following conditions:
\[
z_j^0 \in\bbr^{d+1},\quad \Omega \in\bbr^{(d+1)\times (d+1)}, \quad j \in {\mathcal N},
\]
then $(Z,K, \Lambda)$ is a solution to system \eqref{B-4} with following conditions:
\[
\gamma=\gamma_1,\quad \mu=\mu_1,\quad \Gamma=\Gamma_1,\quad \kappa_{ij}^0=\kappa_{ji}^0.
\]
This implies that system \eqref{B-2} can be reduced to \eqref{B-4} with real natural frequency matrices. Then, Subsystem A \eqref{B-2} reduces to the LS model with adaptive couplings:
\begin{align}\label{B-4}
\begin{cases}
\dot{x}_j=\Omega x_j +\displaystyle\frac{1}{N}\sum_{k=1}^N\kappa_{jk} \Big(\langle x_j, x_j \rangle x_k -\langle x_k, x_j \rangle x_j \Big),\quad t>0,\quad j \in {\mathcal N}, \vspace{0.2cm} \\
\dot{\kappa}_{jk}=-\gamma\kappa_{jk}+\mu\Gamma(x_j, x_k), \vspace{0.2cm} \\
(x_i(0), \kappa_{jk}(0)) = (x_i^0, \kappa_{jk}^0) \in \bbs^d \times \bbr_+.
\end{cases}
\end{align}

\vspace{0.2cm}

Next, we show that  Subsystem A and Subsystem B can be reduced to the Kuramoto model with adaptive couplings in three different ways:
\begin{align}\label{B-6}
\begin{cases}
\dot{\theta}_j =\nu_j +\displaystyle\frac{1}{N}\sum_{l=1}^N\kappa_{jl}\sin(\theta_l-\theta_j),\\
\dot{\kappa}_{jk}=-\gamma\kappa_{jk}+\mu\Gamma(\theta_k -\theta_j),\quad t>0,\quad j, k \in {\mathcal N},
\end{cases}
\end{align}
where $\Gamma$ satisfies
\[
\Gamma(-\theta)=\Gamma(\theta),\quad \Gamma(\theta+2\pi)=\Gamma(\theta).
\]
\vspace{0.2cm}

\noindent First, let $\{x_j \}$ be a solution of the LS model \eqref{B-4} with adaptive couplings. We set 
\[
d=1,\quad x_j =\begin{bmatrix}
\cos\theta_j \\
\sin\theta_j
\end{bmatrix},\quad
\Omega_j =\begin{bmatrix}
0&-\nu_j \\
\nu_j &0
\end{bmatrix}.
\]
Then, system \eqref{B-4}$_1$ can be converted to
\[
\dot{\theta}_j \begin{bmatrix}
-\sin\theta_j \\
\cos\theta_j
\end{bmatrix}=
\nu_j \begin{bmatrix}
-\sin\theta_j \\
\cos\theta_j
\end{bmatrix}+\frac{1}{N}\sum_{k=1}^N\kappa_{jk}\left(\begin{bmatrix}
\cos\theta_k\\
\sin\theta_k
\end{bmatrix}-\cos(\theta_j-\theta_i)\begin{bmatrix}
\cos\theta_j \\
\sin\theta_j
\end{bmatrix}
\right).
\] 
This yields
\begin{equation} \label{B-6-1}
\dot{\theta}_j =\nu_j +\frac{1}{N}\sum_{k=1}^N\kappa_{jk}\sin(\theta_k-\theta_j)
\end{equation}
and we can obtain \eqref{B-6}$_1$. \newline

Note that the dynamics $\eqref{B-4}_2$ of $\kappa_{jk}$  can also be expressed as
\[
\dot{\kappa}_{jk}=-\gamma\kappa_{jk}+\mu\Gamma\left( \begin{bmatrix}
\cos\theta_j\\
\sin\theta_j
\end{bmatrix},\begin{bmatrix}
\cos\theta_k\\
\sin\theta_k
\end{bmatrix}\right).
\]
From the simple assumption $\Gamma(x, y)= {\tilde \Gamma}(\|x-y\|)$, there is a proper function $\hat{\Gamma}$ with the following properties:
\begin{align}\label{B-7}
\Gamma\left( \begin{bmatrix}
\cos\theta_i\\
\sin\theta_i
\end{bmatrix},\begin{bmatrix}
\cos\theta_j\\
\sin\theta_j
\end{bmatrix}\right)= \tilde{\Gamma} \left(2\left|\sin\left(\frac{\theta_i-\theta_j}{2}\right)\right|\right)= \hat{\Gamma}(\theta_i-\theta_j).
\end{align}
It is easy to check that $\hat{\Gamma}$ satisfies
\[
\hat{\Gamma}(-\theta)=\hat{\Gamma}(\theta),\quad \hat{\Gamma}(\theta+2\pi)= \hat{\Gamma}(\theta).
\]
Thus, system $\eqref{B-6}_2$ becomes 
\begin{equation} \label{B-6-2}
\dot{\kappa}_{jk}=-\gamma\kappa_{jk}+\mu \hat{\Gamma}(\theta_k-\theta_j).
\end{equation}
Finally, we combine \eqref{B-6-1} and \eqref{B-6-2} to derive the Kuramoto model with adaptive couplings \eqref{B-6}. \newline

Second, we consider Subsystem A \eqref{B-2}. Let $(Z, K)$ be a solution to \eqref{B-2} with
\begin{align}\label{B-8}
d=0,\quad z_j=e^{\mathrm{i}\theta_j},\quad \Omega_j=\mathrm{i}\nu_j.
\end{align}
Then we substitute \eqref{B-8} into $\eqref{B-2}_1$ to get 
\begin{align*}
\mathrm{i}\dot{\theta}_je^{\mathrm{i}\theta_j}=\mathrm{i}\nu_je^{\mathrm{i}\theta_j}+\frac{1}{N}\sum_{k=1}^N\kappa_{jk}\left(e^{\mathrm{i}\theta_k}-e^{\mathrm{i}(2\theta_j-\theta_k)}\right)
\end{align*}
which can be simplified as 
\[
\dot{\theta}_j=\nu_j+\frac{2}{N}\sum_{k=1}^N\kappa_{jk} \sin(\theta_k-\theta_j).
\]
By the same arguments as in \eqref{B-7}, we can also reduce \eqref{B-2}$_2$ to 
\[
\dot{\kappa}_{jk}=-\gamma_0\kappa_{jk}+\mu_0 \hat{\Gamma}_0(\theta_k-\theta_j).
\]
Again, Subsystem A can be reduced to the Kuramoto model with adaptive couplings. \newline

Third, we consider Subsystem B \eqref{B-3}. Let $(Z, K)$ be a solution to \eqref{B-3} with \eqref{B-8}. Then, by  the same argument as in Subsystem A, we can convert \eqref{B-3} as 
\begin{align*}
\begin{cases}
\dot{\theta}_j=\nu_j+\displaystyle\frac{2}{N}\sum_{k=1}^N\kappa_{jk}\sin(\theta_k-\theta_j),\\
\dot{\kappa}_{jk} =-\gamma\kappa_{jk} +\mu_1 \hat{\Gamma}(\theta_k -\theta_j).
\end{cases}
\end{align*}
This implies that Subsystem B can be reduced to the Kuramoto model with adaptive coupling gains. To sum up,  we can visualize aforementioned discussions in the following diagram.
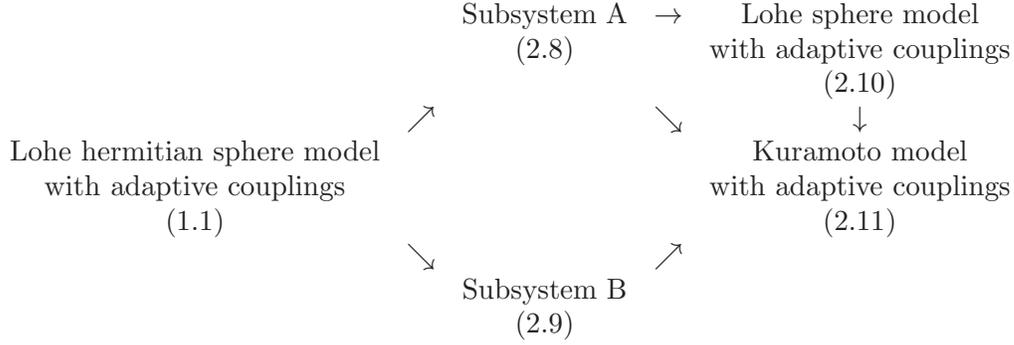
\vspace{0.5cm}
\begin{center}
\begin{tabular}{c c c c c}
&&Subsystem A  &$\rightarrow$&Lohe sphere model\\
&&\eqref{B-2}&&with adaptive couplings\\
&&&&\eqref{B-4}\\
&$\nearrow$&&$\searrow$&$\downarrow$\\
Lohe hermitian sphere model &&&&Kuramoto model \\
with adaptive couplings&&&&with adaptive couplings\\
\eqref{A-0}&&&&\eqref{B-6}\\
&$\searrow$&&$\nearrow$&\\
&&Subsystem B\\
&& \eqref{B-3}
\end{tabular}
\end{center}
\vspace{0.5cm}

The LHS model with adaptive coupling gains can be reduced to Subsystem A and Subsystem B by setting $\lambda_{jk}\equiv0$ and $\kappa_{jk}\equiv0$, respectively. Each subsystem can  also be reduced to the Kuramoto model with adaptive couplings. This implies that each coupling term of the LHS model with adaptive couplings can be reduced to the Kuramoto model  with adaptive couplings. So we can conclude that the LHS model \eqref{A-0} with adaptive couplings is well-defined. 

\subsection{From the Stuart-Landau model to the LHS model}\label{sec:2.3}
In this subsection, we explain how the special coupling gain relation $\lambda_{jk} = -\frac{1}{2} \kappa_{jk}$ can arise in the reduction from  the generalized Stuart-Landau model  to the LHS model. \newline

Consider a {\it generalized Stuart-Landau model} on $\bbc^{d+1}$:  
\begin{equation} \label{B-9}
\frac{dz_j}{dt} =\big((1-\|z_j\|^2)I_{d+1}+\Omega \big)z_j+\frac{\kappa}{N}\sum_{k=1}^N(z_k-z_j),
\end{equation}
where $z_j\in\bbc^{d+1}$ for all $j\in\mathcal{N}$, $\Omega$ is a skew-hermitian matrix with the size $(d+1)\times (d+1)$ and $I_{d+1}$ is the identity matrix with the size $(d+1)\times (d+1)$. \newline

We substitute the ansatz: 
\[
z_j=r_jw_j, \quad r_j=\|z_j\| \quad \mbox{and} \quad w_j=\frac{z_j}{\|z_j \|}, \quad \forall~j \in \mathcal{N}
\]
into \eqref{B-9} to see 
\begin{equation} \label{B-10}
\dot{r}_jw_j+r_j\dot{w}_j=(1-r_j^2)r_jw_j+r_j\Omega w_j+\frac{\kappa}{N}\sum_{k=1}^N(r_k w_k -r_jw_j).
\end{equation}
Then, $\langle w_j, \eqref{B-10} \rangle$ implies
\begin{align}\label{B-11}
\dot{r}_j+r_j\langle w_j, \dot{w}_j\rangle=(1-r_j^2)r_j+r_j\langle w_j,  \Omega w_j\rangle+\frac{\kappa}{N}\sum_{k=1}^N (r_k \langle w_j, w_k \rangle-r_j).
\end{align}
If we take the real part of \eqref{B-11}, one has
\begin{align}\label{B-12}
\dot{r}_j=(1-r_j^2)r_j+\frac{\kappa}{N}\sum_{k=1}^N(r_k \mathrm{Re}(\langle w_j, w_k \rangle)-r_j).
\end{align}
Here we used the relations:
\[ \langle w_j, \dot{w}_j\rangle+\langle \dot{w}_j, w_j\rangle=   \langle w_j, \dot{w}_j\rangle + \overline{\langle w_j, \dot{w}_j\rangle} = 2 \mbox{Re}~ \langle w_j, \dot{w}_j\rangle, \quad 
\langle w_j,  \Omega w_j\rangle = 0. \]
Now, we combine \eqref{B-11} and \eqref{B-12} to get 
\begin{align*}
\dot{w}_j=\Omega w_j+\frac{\kappa}{N}\sum_{k=1}^N\frac{r_k}{r_j}\Big(w_k -\mathrm{Re}(\langle w_k, w_j\rangle) w_j \Big).
\end{align*}
Similarly, we impose $r_j\equiv1$ on \eqref{B-12} to obtain
\begin{align}
\begin{aligned} \label{B-12-1}
\dot{w}_j&=\Omega w_j+\frac{\kappa}{N}\sum_{k=1}^N \Big[ w_k - \frac{1}{2} \Big(\langle w_k, w_j\rangle+\langle w_j, w_k \rangle \Big )w_j \Big] \\
&=\Omega w_j+\frac{\kappa}{N}\sum_{k=1}^N(w_k -\langle w_k, w_j\rangle w_j)-\frac{\kappa}{2N}\sum_{k=1}^N(\langle w_j, w_k \rangle-\langle w_k, w_j\rangle)w_j \\
& = \Omega w_j+\frac{\kappa}{N}\sum_{k=1}^N\Big[ w_k -\frac{1}{2} \Big(\langle w_k, w_j\rangle+\langle w_j, w_k \rangle \Big)w_j \Big].
\end{aligned}
\end{align}
Note that this is the special case of the LHS model \eqref{A-0} with $\kappa_1=-\frac{\kappa_0}{2}$. \newline

Next, we show that system \eqref{B-12-1} can be embedded as a system on the Euclidean space by extending $(d+1)$-dimensional complex-valued vector $w\in\bbc^{d+1}$ to $2(d+1)$-dimensional real-valued vector $\tilde{w}\in\bbr^{2(d+1)}$ with the following map:
\[
w=(w^1, \cdots, w^{d+1}) \quad \mapsto \quad \tilde{w}=\big(\mathrm{Re}(w^1), \cdots, \mathrm{Re}(w^{d+1}),  \mathrm{Im}(w^1), \cdots, \mathrm{Im}(w^{d+1})\big).
\] 
Now we will rewrite \eqref{B-12-1} in terms of  $\{\tilde{w}_j\}$. First, it is easy to see that
\begin{equation} \label{B-13}
\dot{\tilde{w}}_j=\tilde{\dot{w}}_j.
\end{equation}
By simple calculation, we have
\begin{align*}
\Omega w_j&=\big(\mathrm{Re}(\Omega )+\mathrm{i}\mathrm{Im}(\Omega)\big)\big(\mathrm{Re}(w_j)+\mathrm{i}\mathrm{Im}(w_j)\big)\\
&=\big(\mathrm{Re}(\Omega)\mathrm{Re}(w_j)-\mathrm{Im}(\Omega)\mathrm{Im}(w_j)\big)+\mathrm{i}\big(\mathrm{Im}(\Omega)\mathrm{Re}(w_j)+\mathrm{Re}(\Omega)\mathrm{Im}(w_j)\big).
\end{align*}
This yields
\begin{align}
\begin{aligned} \label{B-14}
& \mathrm{Re}(\Omega w_j)=\mathrm{Re}(\Omega)\mathrm{Re}(w_j)-\mathrm{Im}(\Omega)\mathrm{Im}(w_j), \\
& \mathrm{Im}(\Omega w_j)=\mathrm{Im}(\Omega)\mathrm{Re}(w_j)+\mathrm{Re}(\Omega)\mathrm{Im}(w_j).
\end{aligned}
\end{align}
Since $\Omega$ is a $(d+1)\times (d+1)$ complex skew-hermitian matrix, we know that $\mathrm{Re}(\Omega)$ and $\mathrm{Im}(\Omega)$ are symmetric. From this, we can define $2(d+1)\times 2(d+1)$ skew-symmetric matrix $\tilde{\Omega}$ as follows:  
\[
\tilde{\Omega}=\begin{bmatrix}
\mathrm{Re}(\Omega)&-\mathrm{Im}(\Omega)\\
\mathrm{Im}(\Omega)&\mathrm{Re}(\Omega)
\end{bmatrix}.
\] 
Then we have
\begin{equation} \label{B-15}
\widetilde{\Omega w_j}=\tilde{\Omega} \tilde{w}_j.
\end{equation}
Next, we rewrite $\langle w_k, w_j\rangle+\langle w_j, w_k \rangle$ in terns of $\tilde{w}_k$ and $\tilde{w}_j$ as follows. By definition of the complex inner-product, we have
\begin{align}
\begin{aligned} \label{B-16}
\langle w_k, w_j\rangle+\langle w_j, w_k \rangle&=w_k^\dagger w_j+w_j^\dagger w_k =2\mathrm{Re}(w_k)^T \mathrm{Re}(w_j)+2\mathrm{Im}(w_k)^T\mathrm{Im}(w_j)\\
&=2\tilde{w}_k^T\tilde{w}_j=2\langle \tilde{w}_k, \tilde{w}_j\rangle.
\end{aligned}
\end{align}
Finally we can express system \eqref{B-12-1} with $\{\tilde{w}_j\}$ and $\Omega$ using \eqref{B-13}, \eqref{B-14}, \eqref{B-15} and \eqref{B-16} to get 
\[
\dot{\tilde{w}}_j=\tilde{\Omega}\tilde{w}_j+\frac{\kappa}{N}\sum_{k=1}^N\left(\tilde{w}_k -\langle \tilde{w}_k,\tilde{w}_j\rangle\tilde{w}_j\right)
\]
which is exactly the Lohe sphere model. In summary, from the proper map between $\bbc^{d+1}$ and $\bbr^{2(d+1)}$, we can transform the special case of the LHS model \eqref{A-0} with $\lambda_{jk}=-\frac{\kappa_{jk}}{2}$ to the LS model. Thus, we can see that system \eqref{B-13} is a gradient flow as in  the LS model (see Proposition 5.1 in  \cite{H-Ko-R}).

\section{Frameworks for complete aggregation and main results} \label{sec:3}
\setcounter{equation}{0}
In this section, we briefly present our main results and sufficient frameworks leading to complete aggregation in the sense of Definition \ref{D1.1}. As noted in the previous section, we consider four different cases depending on the relations between coupling gains $\kappa_{jk},~\lambda_{jk}$ and coupling law $\Gamma_0$ (anti-Hebbian or Hebbian law). \newline

\noindent $\diamond$~(Coupling gain pair): Depending on the relation between $\kappa_{jk}$ and $\lambda_{jk}$, we consider the following two cases: \newline
\begin{itemize}
\item
Stuart-Landau coupling gain pair $(\kappa_{jk}, \lambda_{jk})$:
\[  \kappa_{jk} > 0, \quad \lambda_{jk} = -\frac{1}{2}\kappa_{jk}, \quad j, k \in {\mathcal N}. \]
\vspace{0.1cm}
\item
General coupling gain pair $(\kappa_{jk}, \lambda_{jk})$:
\[ \kappa_{jk} > 0 \quad \mbox{and} \quad \lambda_{jk} \in \bbr, \quad j,k  \in {\mathcal N}. \]
\end{itemize}
\vspace{0.2cm}

\noindent $\diamond$~(Coupling law $\Gamma_0$):~Consider anti-Hebbian and Hebbian laws:
\begin{equation} \label{C-0}
\Gamma_0(w, z) = 
\begin{cases}
\displaystyle \|w-z\|^2, \quad & \mbox{Anti-Hebbian law} \\
\displaystyle 1-\frac{1}{2}\|w-z\|^2, \quad & \mbox{Hebbian law} 
\end{cases}
\end{equation}

The motivation for \eqref{C-0} can be explained as follows. In literature \cite{HKLN, HLLP, HNP} on the synchronization with an adaptive coupling law, the following coupling law 
\begin{equation} \label{C-0-1}
\Gamma(\theta, {\tilde \theta}) = \cos(\theta - \tilde{\theta}) 
\end{equation}
was often employed. Note that when the phase difference between the interactiong oscillators is small, it increases the mutual coupling strength. Thus, it is called Hebbian coupling law. In contrast, when the phase difference is small, there is a case in which the rate of increment in coupling strength becomes small. This is called ``{\it anti-Hebbian coupling law}" and the ansatz:
\begin{equation} \label{C-0-2}
\Gamma(\theta, {\tilde \theta}) =  |\sin(\theta - {\tilde \theta})|  
\end{equation}
was used in aforementioned literature. Note that on $\bbh \bbs^d$, 
\[ \Gamma_0(z, {\tilde z}) = \|z - \tilde{z} \|^2 = 2(1 - \mathrm{Re}\langle z, {\tilde z} \rangle) = 2 \left(1 - \mathrm{Re}\big(\cos \theta(z, \tilde z)\big)\right),\]
where $\theta(z, \tilde z)$ is the angle between $z$ and $\tilde z$, and $\Gamma_0$ becomes smaller when the angle is small. In this sense, it plays the same role as anti-Hebbian law \eqref{C-0-2}. In contrast, real part of $\Gamma_0(z, {\tilde z}) = 1 - \frac{1}{2} \|z - \tilde{z} \|^2$ exhibits the same dynamics as \eqref{C-0-1}. 

\subsection{SL coupling gain pair}  \label{sec:3.1}
Consider the Stuart-Landau coupling gain pair:
\begin{equation} \label{C-1-0}
\kappa_{jk} > 0, \quad \lambda_{jk} = -\frac{1}{2}\kappa_{jk}, \quad\forall~ t \geq 0, \quad j, k \in {\mathcal N}. 
\end{equation}
In fact, one can show that once initial gain pair satisfy \eqref{C-1-0}, then the relation \eqref{C-1-0} will be propagated along \eqref{A-1} under suitable conditions on system parameters and coupling laws (see Lemma \ref{L4.1}). 
\subsubsection{Anti-Hebbian coupling law} \label{sec:3.1.1} Consider the anti-Hebbian coupling law: 
\begin{equation} \label{C-1-1}
\Gamma_0(z,{\tilde z})=\|z - {\tilde z}\|^2. 
\end{equation}
Under the setting \eqref{C-1-0} and \eqref{C-1-1}, system \eqref{A-1} becomes 
\begin{equation} \label{C-1-2}
\begin{cases}
\displaystyle \dot{z}_j = \frac{1}{N}\sum_{k=1}^N \kappa_{jk}\left[ z_k-\frac{1}{2}\Big(\langle z_j, z_k \rangle+\langle z_k, z_j \rangle \Big)z_j\right ], \quad t > 0, \vspace{0.2cm} \\
\displaystyle \dot{\kappa}_{jk} = -\gamma \kappa_{jk} + \mu \| z_j - z_k \|^2, \vspace{0.2cm} \\
\displaystyle \displaystyle (z_j, \kappa_{jk})(0) =(z_j^0, \kappa_{jk}^0) \in \bbh\bbs^{d} \times \bbr_+, \quad  j, k \in \mathcal{N}.
\end{cases}
\end{equation}
For the emergent estimate, we use a Lyapunov functional approach: for $i, j \in {\mathcal N}$, 
\begin{equation} \label{C-1-3}
\mathcal{L}_{ij}=\frac{1}{2}\|z_i-z_j\|^2+\frac{1}{4\mu N}\sum_{k=1}^{N}(\kappa_{ik} -\kappa_{jk})^2.
\end{equation}
Note that at the completely aggregated state
\[ z_i = z, \quad \kappa_{ij} = \kappa, \quad i, j \in {\mathcal N} \]
the functional $\mathcal{L}_{ij}$ is exactly zero. Thus, we can see that the functional $\mathcal{L}_{ij}$ can measure how a state configuration and coupling gains are close to complete aggregated state. \newline

Now, we state our first result on the emergent dynamics for \eqref{A-0}.  
\begin{theorem} \label{T3.1} 
Suppose initial data $(Z^0, K^0)$ satisfy
\begin{equation} \label{C-1-3-1}
 \max_{i,j} \mathcal{L}_{ij}^0 <1, 
 \end{equation}
and let $(Z,K)$ be a solution to \eqref{C-1-2}. Then, one has 
\[
\lim_{t \to \infty} \| z_i(t)-z_j(t) \| = 0 \quad \text{and} \quad \lim_{t \to \infty} \kappa_{ij}(t) = 0, \quad i, j \in {\mathcal N}.
\]
\end{theorem}
\begin{proof} We leave its proof in Section \ref{sec:4}.
\end{proof}
\subsubsection{Hebbian coupling law} \label{sec:3.1.2} In this part, we consider the Hebbian law:
\begin{equation} \label{C-1-4}
\Gamma_0(z, {\tilde z})=1-\frac{\|z -{\tilde z} \|^2}{2}.
\end{equation}
Under the setting \eqref{C-1-0} and \eqref{C-1-4}, system \eqref{A-0} becomes 
\begin{align}\label{C-2-0}
\begin{cases}
\displaystyle  \dot{z}_j = \frac{1}{N}\sum_{k=1}^N \kappa_{jk}\Big[ z_k-\frac{1}{2}\Big( \langle z_j, z_k \rangle+\langle z_k, z_j \rangle\big)z_j \Big], \quad  t > 0, \vspace{0.2cm} \\
\displaystyle \dot{\kappa}_{jk} = -\gamma \kappa_{jk} + \mu \left(1-\displaystyle\frac{\|z_k -z_j \|^2}{2}\right),  \vspace{0.2cm} \\
\displaystyle \displaystyle (z_j, \kappa_{jk})(0) =(z_j^0, \kappa_{jk}^0) \in \bbh\bbs^{d} \times \bbr_+, \quad j, k \in \mathcal{N}.
\end{cases}
\end{align}
For the emergent dynamics of \eqref{C-2-0}, we introduce a Lyapunov function:
\begin{equation*}
{\mathcal D}_{ij}(Z):= \frac{1}{2}\|z_i-z_j\|^2, \qquad {\mathcal D}(Z):= \max_{i,j} {\mathcal D}_{ij}(Z).
\end{equation*}
Note that  ${\mathcal D}$ is the half of the square of state diameter, and complete aggregation occurs if 
\[
\lim_{t \to \infty} {\mathcal D}(Z(t))=0.
\]
The maximum of differentiable functions does not need to be differentiable, hence we cannot guarantee differentiability of ${\mathcal D}(Z)$. However, it follows from the analyticity of each ${\mathcal D}_{ij}(Z)$, ${\mathcal D}(Z)$ is differentiable almost everywhere and we can regard $\dot{{\mathcal D}}(Z(t))$ as a weak derivative of ${\mathcal D}(Z)$. By the continuity and estimate for $\dot{{\mathcal D}}(Z)$ a.e. is enough to derive the estimate for ${\mathcal D}(Z)$ via direct integration. \newline

Now, we present our second result as follows.
\begin{theorem} \label{T3.2}
Suppose there exist a constant $\kappa$ such that
\begin{equation} \label{C-2-4}
0 < \kappa < \min \left\{  \frac{\mu}{\gamma},~\min_{i,j}\kappa_{ij}^0 \right \}, \quad \max \left\{  \max_{i,j} \kappa^0_{ij}, ~\frac{\mu}{\gamma} \right \} \leq \frac{2\mu \kappa}{2\mu-\gamma \kappa}, 
\quad {\mathcal D}(Z^0) < 1-\frac{\gamma}{\mu} \kappa.
\end{equation}
and let $(Z,K)$ be a solution to \eqref{C-2-0}. Then, there exist positive constants $C_0>0$ and $C_1>0$ satisfying
\[
{\mathcal D}(Z(t)) \leq C_0 e^{-C_1 t}, \quad t>0.
\]
\end{theorem}

\begin{proof}
We leave its proof in Section \ref{sec:4.2}.
\end{proof}

\subsection{Asymptotically SL coupling gain pair} \label{sec:3.2}
In this subsection, we consider a coupling gain pair:
\[ \kappa_{jk} > 0, \quad \lambda_{jk} \in \bbr, \quad\forall~ j, k \in {\mathcal N}. \]
Note that unlike the previous subsection, we do not assume any functional relation between $\kappa_{jk}$ and $\lambda_{jk}$. Nevertheless, we can still rewrite \eqref{A-1} as a perturbation of \eqref{A-3}. To see this, we recall the dynamics of $z_j$:
\begin{equation} \label{E-1}
 \dot{z}_j= \frac{1}{N}\sum_{k=1}^N\kappa_{jk}\big(\langle z_j, z_j\rangle z_k-\langle z_k, z_j\rangle z_j\big)+\frac{1}{N}\sum_{k=1}^N\lambda_{jk}\big(\langle z_j, z_k\rangle-\langle z_k,z_j\rangle\big)z_j.
\end{equation} 
To use the result in Section \ref{sec:3.1}, we set
\begin{equation} \label{E-1-1}
\tilde{\lambda}_{jk}:=\frac{1}{2}\kappa_{jk}+\lambda_{jk}, \quad\forall~ j,k \in \mathcal{N}.
\end{equation}
Then, $\tilde{\lambda}_{jk} = 0$ corresponds to exactly the same situation in Section \ref{sec:3.1}. Now we can rewrite \eqref{E-1} using \eqref{E-1-1} into
\begin{equation*}
\dot{z}_{j} = \frac{1}{N}\sum_{k=1}^{N} \kappa_{jk}(z_k-R_{kj}z_j)+ \frac{1}{N}\sum_{k=1}^{N} \tilde{\lambda}_{jk}(h_{jk}-h_{kj})z_j,
\end{equation*}
and the dynamics of $\tilde{\lambda}_{jk}$ can also be expressed as follows:
\begin{equation} \label{E-1-3}
\dot{\tilde{\lambda}}_{jk} =\frac{1}{2}\dot{\kappa}_{jk}+\dot{\lambda}_{jk} =\frac{1}{2} \Big(-\gamma_0\kappa_{jk}+\mu_0\Gamma_0(z_j, z_k) \Big)+ \Big(-\gamma_1\lambda_{jk}+\mu_1\Gamma_1(z_j, z_k) \Big).
\end{equation}
In what follows, we set
\begin{align}\label{E-1-4}
\gamma_0=\gamma_1=\gamma,\quad \mu_0=\mu_1=\mu.
\end{align}
We combine \eqref{E-1-3} and \eqref{E-1-4} to get 
\begin{equation*}
\dot{\tilde{\lambda}}_{jk} =-\gamma\left(\frac{1}{2}\kappa_{jk}+\lambda_{jk}\right)+\mu\left(\frac{1}{2}\Gamma_0(z_j, z_k)+\Gamma_1(z_j, z_k)\right).
\end{equation*}
Now, we set
\[
\tilde{\Gamma}(z, w) :=\frac{1}{2}\Gamma_0(z, w)+\Gamma_1(z, w)
\]
to rewrite
\begin{equation} \label{E-1-6}
\dot{\tilde{\lambda}}_{jk}=-\gamma \tilde{\lambda}_{jk}+\mu\tilde{\Gamma}(z_j, z_k).
\end{equation}
Finally we combine \eqref{E-1} and \eqref{E-1-6} to get 
\begin{equation}\label{E-1-7}
\begin{cases}
\displaystyle \dot{z}_j=  \frac{1}{N}\sum_{k=1}^N \kappa_{jk}\Big[ z_k-\frac{1}{2}\Big( \langle z_j, z_k \rangle+\langle z_k, z_j \rangle\big)z_j \Big] +\frac {1}{N}\sum_{k=1}^N\tilde{\lambda}_{jk}(\langle z_j, z_k\rangle-\langle z_k, z_j\rangle)z_j,~~t > 0, \vspace{0.2cm} \\
\displaystyle \dot{\kappa}_{jk}=-\gamma\kappa_{jk}+\mu\Gamma_0(z_j,z_k),\quad \dot{\tilde{\lambda}}_{jk}=-\gamma\tilde{\lambda}_{jk}+\mu\tilde{\Gamma}(z_j, z_k),\quad j, k \in \mathcal{N}, \vspace{0.2cm}  \\
\displaystyle (z_j, \kappa_{jk}, {\tilde \lambda}_{jk})(0) =(z_j^0, \kappa_{jk}^0, {\tilde \lambda}_{jk}^0) \in \bbh\bbs^{d} \times \bbr_+ \times \bbr.
\end{cases}
\end{equation}
In what follows, we consider only following cases:
\[ {\tilde \lambda}_{ij}~\mbox{is independent of $i$ and $j$, but is a function of $t$, and $\tilde \Gamma \equiv 0$}. \]
\subsubsection{Anti-Hebbian coupling law} \label{sec:3.2.1}
In this part, we study emergent dynamics of \eqref{E-1-7} with anti-Hebbian coupling law:  
\begin{equation}\label{E-2-0-0}
\begin{cases}
\displaystyle \dot{z}_j=  \frac{1}{N}\sum_{k=1}^N \kappa_{jk}\Big[ z_k-\frac{1}{2}\Big( \langle z_j, z_k \rangle+\langle z_k, z_j \rangle\big)z_j \Big] +\frac{1}{N}\sum_{k=1}^N\tilde{\lambda}_{jk}(\langle z_j, z_k\rangle-\langle z_k, z_j\rangle)z_j, \vspace{0.2cm} \\
\displaystyle \dot{\kappa}_{jk}=-\gamma\kappa_{jk}+\mu\|z_j-z_k\|^2,\quad \dot{\tilde{\lambda}}_{jk}=-\gamma\tilde{\lambda}_{jk}+\mu\tilde{\Gamma}(z_j, z_k), \vspace{0.2cm} \\
\displaystyle (z_j, \kappa_{jk}, {\tilde \lambda}_{jk})(0) =(z_j^0, \kappa_{jk}^0, {\tilde \lambda}_{jk}^0) \in \bbh\bbs^{d} \times \bbr_+ \times \bbr, \quad j, k \in \mathcal{N}.
\end{cases}
\end{equation}
Our third main result is concerned with a sufficient framework leading to complete aggregation. 

\begin{theorem} \label{T3.3}
Suppose that the following relations hold
\[
\tilde{\lambda}_{ij}^0={\tilde \lambda}^0, \quad\forall~ i,j \in \mathcal{N} \quad \text{and} \quad \tilde{\Gamma}(t) \equiv 0, \quad t>0,
\]
for some constant $\tilde\lambda^0$, and let $(Z, K, {\tilde \Lambda})$ be a solution to \eqref{E-2-0-0} with initial data satisfying the following conditions:
\begin{align*}
\quad 1 > \max_{i,j}\frac{2|\tilde\lambda^0|}{\kappa_{ij}^0}+\max_{k,l}\mathcal{L}_{kl}^0,
\end{align*}
then we have
\[
\lim_{t \to \infty} \| z_i(t)-z_j(t) \| = 0 \quad \text{and} \quad \lim_{t \to \infty} \kappa_{ij}(t) = 0.
\]
\end{theorem}
\begin{proof} We leave its proof in Section \ref{sec:5.1}. 
\end{proof}
\begin{remark} Since the coupling gains tend to zero asymptotically, the presented result is completely different from the previous result in \cite{H-P2} in which the coupling gains take the same positive constant:
\[ \kappa_{ij}(t) = \kappa > 0, \quad\forall~ i, j \in {\mathcal N}. \]
\end{remark}

\subsubsection{Hebbian coupling law} \label{sec:3.2.2}
Consider system \eqref{E-1-7} with a Hebbian coupling law:
\begin{align}\label{NE-4}
\begin{cases}
\displaystyle \dot{z}_j=  \frac{1}{N}\sum_{k=1}^N \kappa_{jk}\Big[ z_k-\frac{1}{2}\Big( \langle z_j, z_k \rangle+\langle z_k, z_j \rangle\big)z_j \Big] +\frac{1}{N}\sum_{k=1}^N\tilde{\lambda}_{jk}(\langle z_j, z_k\rangle-\langle z_k, z_j\rangle)z_j, \vspace{0.2cm} \\
\displaystyle \dot{\kappa}_{jk}=-\gamma\kappa_{jk}+\mu \left(\displaystyle 1-\frac{\|z_j -z_k \|^2}{2} \right),\quad \dot{\tilde{\lambda}}_{jk}=-\gamma\tilde{\lambda}_{jk}+\mu\tilde{\Gamma}(z_j, z_k), \vspace{0.2cm} \\
(z_j, \kappa_{jk}, {\tilde \lambda}_{jk})(0) = (z_j^0, \kappa_{jk}^0, {\tilde \lambda}_{jk}^0) \in \bbh\bbs^{d} \times \bbr_+ \times \bbr,\quad j, k\in \mathcal{N},
\end{cases}
\end{align}

\vspace{0.2cm}

Similar to Section \ref{sec:3.1.2}, one has the same emergent dynamics. 
\begin{theorem} \label{T3.4}
Suppose there exist a constant $\kappa$ such that
\begin{equation}\label{Y-0}
2|\tilde{\lambda}^0| < \kappa < \min \left\{  \frac{\mu}{\gamma},~\min_{i,j}\kappa_{ij}^0 \right \}, \quad 
\max \left\{  \max_{i,j} \kappa^0_{ij}, ~\frac{\mu}{\gamma} \right \} \leq \frac{2\mu (\kappa-2|\tilde{\lambda^0}|)}{2\mu-\gamma \kappa}, 
\quad {\mathcal D}(Z^0) < 1-\frac{\gamma}{\mu} \kappa.
\end{equation}
and let $(Z,K)$ be a solution to \eqref{NE-4}. Then, there exist positive constants $C_2>0$ and $C_3>0$ satisfying
\[
{\mathcal D}(Z(t)) \leq C_2 e^{-C_3 t}, \quad t>0.
\]
\end{theorem}

\begin{proof}
We leave its proof in Section \ref{secL5.6:5.2}. 
\end{proof}

\section{Collective dynamics under Stuart-Landau coupling gain pair}\label{sec:4}
\setcounter{equation}{0}
In this section, we study emergent dynamics of system \eqref{A-1} with the initial Stuart-Landau coupling gain pair:
\begin{equation} \label{D-0}
\kappa^0_{ij} > 0, \quad \lambda^0_{ij} = -\frac{1}{2}\kappa^0_{ij}, \quad\forall~ i, j \in {\mathcal N}.
\end{equation}
First, we show that if the initial SL coupling gain pair satisfies \eqref{D-0}, then it is propagated along the dynamics \eqref{A-1}.
\begin{lemma} \label{L4.1}
Suppose that system parameters and initial coupling strengths satisfy
\begin{align}\label{D-0-0}
\gamma_0=\gamma_1=\gamma,\quad \mu_0=\mu_1=\mu,\quad \Gamma_0+2\Gamma_1=0,\quad \lambda_{ij}^0 =-\frac{1}{2}\kappa_{ij}^0,
\end{align}
and let $(Z, K, \Lambda)$ be a solution of \eqref{A-0}. Then we have
\begin{equation} \label{D-0-1}
\lambda_{ij}(t) = -\frac{1}{2}\kappa_{ij}(t), \quad \forall~t \geq 0, \quad i, j\in \mathcal{N}.
\end{equation}
\end{lemma}
\begin{proof}  It follows from \eqref{A-0} and \eqref{D-0-0} that 
\[
\dot{\kappa}_{ij}= -\gamma \kappa_{ij}+\mu \Gamma_0(z_i, z_j), \quad \dot{\lambda}_{ij}= -\gamma \lambda_{ij} -\frac{\mu}{2} \Gamma_0(z_i, z_j), \quad t > 0.
\]
This yields
\[ \frac{d}{dt}(\kappa_{ij}+2\lambda_{ij}) = -\gamma (\kappa_{ij} + 2 \lambda_{ij}).  \]
By integrating the above relation, one has the desired estimate:
\[ (\kappa_{ij}+2\lambda_{ij})(t) = e^{-\gamma t} (\kappa_{ij}^0 +2\lambda_{ij}^0) = 0, \quad t \geq 0. \]
\end{proof}
Next, we substitute \eqref{D-0-1} into \eqref{A-0} to get the dynamics for $(Z, K)$: 
\begin{equation} \label{D-0-2}
\begin{cases}
\displaystyle \dot{z}_j = \frac{1}{N}\sum_{k=1}^N \kappa_{jk}\left[ z_k-\frac{1}{2}\Big(\langle z_j, z_k \rangle+\langle z_k, z_j \rangle \Big)z_j\right ],~~t > 0,
\vspace{0.2cm} \\
\displaystyle \dot{\kappa}_{jk} = -\gamma \kappa_{jk} + \mu \Gamma_0(z_j, z_k), \quad  j, k \in \mathcal{N},  \vspace{0.2cm} \\
\displaystyle (z_j, \kappa_{jk})(0) =(z_j^0, \kappa_{jk}^0) \in \bbh\bbs^{d} \times \bbr_+.
\end{cases}
\end{equation}
In what follows, we consider two coupling laws for $\Gamma_0$ as prototype examples for the anti-Hebbian and the Hebbian couplings between the coupling gain and state:
\begin{equation} \label{D-0-3}
\Gamma_0(z, {\tilde z}) = \|z - \tilde{z} \|^2 \quad \mbox{and} \quad \Gamma_0(z, {\tilde z}) = 1 - \frac{1}{2} \|z - \tilde{z} \|^2. 
\end{equation}

\subsection{Anti-Hebbian coupling law}\label{sec:4.1}
Consider system \eqref{D-0-2} with anti-Hebbian coupling law $\eqref{D-0-3}_1$:
\begin{equation} \label{D-1-0}
\begin{cases}
\displaystyle \dot{z}_j = \frac{1}{N}\sum_{k=1}^N \kappa_{jk}\left[ z_k-\frac{1}{2}\Big(\langle z_j, z_k \rangle+\langle z_k, z_j \rangle \Big)z_j\right ],~~t > 0, \vspace{0.2cm} \\
\displaystyle \dot{\kappa}_{jk} = -\gamma \kappa_{jk} + \mu \|z_j - z_k \|^2, \quad  j, k \in \mathcal{N},  \vspace{0.2cm} \\
\displaystyle (z_j, \kappa_{jk})(0) =(z_j^0, \kappa_{jk}^0) \in \bbh\bbs^{d} \times \bbr_+.
\end{cases}
\end{equation}
To study emergent dynamics of \eqref{D-1-0}, we recall a Lyapunov functional ${\mathcal L}_{ij}$ in \eqref{C-1-3}:
\[ \mathcal{L}_{ij}=\frac{1}{2}\|z_i-z_j\|^2+\frac{1}{4\mu N}\sum_{k=1}^{N}(\kappa_{ik} -\kappa_{jk})^2. \]
On $\bbh\bbs^d$, the functional ${\mathcal L}_{ij}$ can be rewritten as follows:
\begin{equation} \label{D-1-0-0}
\mathcal{L}_{ij} = 1 - \mbox{Re}~\langle z_i, z_j \rangle +  \frac{1}{4\mu N}\sum_{k=1}^{N}(\kappa_{ik} -\kappa_{jk})^2. 
\end{equation}
Thus, it is natural to study the time-evolution of $\langle z_i, z_j \rangle$. For notational simplicity, we use  
\begin{equation} \label{D-1-1-1}
h_{ij} := \langle z_i, z_j \rangle, \quad R_{ij} := \mathrm{Re}~h_{ij} = \frac{1}{2}(h_{ij}+h_{ji}), \quad I_{ij} := \mathrm{Im}~h_{ij} = \frac{1}{2\mathrm{i}}(h_{ij}-h_{ji}).
\end{equation}
Then, it is easy to see 
\[ R_{ii} = 1, \quad I_{ii} = 0, \quad |R_{ij} | \leq |h_{ij}| \leq 1, \quad  R_{ij} = R_{ji} \quad \mbox{and} \quad I_{ij} = -I_{ji}, \quad i, j \in {\mathcal N}. \]
We can rewrite \eqref{D-1-0} and a Lyapunov functional in \eqref{D-1-0-0}:
\[
\dot{z}_j = \frac{1}{N}\sum_{k=1}^N \kappa_{jk} (z_k - R_{jk}z_j ), \quad \mathcal{L}_{ij}=1-R_{ij}+\frac{1}{4\mu N}\sum_{k=1}^{N}(\kappa_{ik} -\kappa_{jk})^2.
\]
To sum up, system \eqref{D-1-0} on $\mathbb{HS}^d$ becomes
\begin{align}\label{D-1-2}
\begin{cases}
\displaystyle \dot{z}_j = \frac{1}{N}\displaystyle\sum_{k=1}^N \kappa_{jk} (z_k - R_{jk}z_j ), \quad t > 0, \vspace{0.2cm} \\
\displaystyle \dot{\kappa}_{jk} = -\gamma \kappa_{jk} + \mu \| z_j - z_k \|^2, \quad j, k \in {\mathcal N},  \vspace{0.2cm}  \\
 (z_j, \kappa_{jk})(0) =(z_j^0, \kappa_{jk}^0) \in \bbh\bbs^{d} \times \bbr_+.
\end{cases}
\end{align}
Next, we study the time-evolution of ${\mathcal L}_{ij}$ in a series of lemmas. 
\begin{lemma}\label{L4.2} 
Let $(Z, K)$ be a solution to \eqref{D-1-2}. Then $\mathcal{L}_{ij}$ satisfies
\begin{align*}
\frac{d}{dt}\mathcal{L}_{ij} = -\frac{1}{N}\sum_{k=1}^N(\kappa_{ik}R_{ik}+\kappa_{jk}R_{jk})(1-R_{ij})-\frac{\gamma}{2\mu N}\sum_{k=1}^N(\kappa_{ik}-\kappa_{jk})^2, \quad t > 0.
\end{align*}
\begin{proof}
By direct calculations, one has 
\begin{equation} \label{D-1-3}
\frac{d}{dt}\mathcal{L}_{ij} = -\frac{1}{2}(\dot{h}_{ij}+\dot{h}_{ji})+\frac{1}{2 \mu N}\sum_{k=1}^N(\kappa_{ik}-\kappa_{jk})(\dot{\kappa}_{ik}-\dot{\kappa}_{jk}).
\end{equation}
Note that the terms in the R.H.S. of \eqref{D-1-3} can be estimated as follows. \newline

\noindent $\bullet$~(Estimate of the first term in \eqref{D-1-3}): We use \eqref{D-1-1-1}  and \eqref{D-1-2} to find
\begin{align}
\begin{aligned} \label{D-1-4}
\dot{h}_{ij}+\dot{h}_{ji} =& \langle z_i, \dot{z}_j \rangle + \langle \dot{z}_i, z_j \rangle + \langle z_j, \dot{z}_i \rangle + \langle \dot{z}_j, z_i \rangle\\
=& \frac{2}{N}\sum_{k=1}^N\mathrm{Re}( \langle z_i, \dot{z}_j \rangle + \langle \dot{z}_i, z_j \rangle )\\
=& \frac{2}{N}\sum_{k=1}^N\mathrm{Re}\Big( \kappa_{jk}(h_{ik}-R_{jk}h_{ij}  )+\kappa_{ik}(h_{kj}-R_{ik}h_{ij}  )\Big)\\
=& \frac{2}{N}\sum_{k=1}^N\Big( \kappa_{jk}(R_{ik}-R_{jk}R_{ij}  )+\kappa_{ik}(R_{jk}-R_{ik}R_{ji}  )\Big)\\
=& \frac{2}{N}\sum_{k=1}^N\Big((\kappa_{ik}-\kappa_{jk})(R_{jk}-R_{ik})+(R_{ik}\kappa_{ik}+R_{jk}\kappa_{jk})(1-R_{ij}) \Big).
\end{aligned}
\end{align}

\vspace{0.5cm}

\noindent $\bullet$~(Estimate of the second term in \eqref{D-1-3}): Similar to the first term, one has 
\begin{align}
\begin{aligned} \label{D-1-5}
&\frac{1}{2\mu N}\sum_{k=1}^N(\kappa_{ik}-\kappa_{jk}) (\dot{\kappa}_{ik}-\dot{\kappa}_{jk}) \\
& \hspace{0.5cm} = \frac{1}{2\mu N}\sum_{k=1}^N (\kappa_{ik}-\kappa_{jk})(-\gamma\kappa_{ik}-\mu h_{ik}-\mu h_{ki} +\gamma\kappa_{jk}+\mu h_{jk} + \mu h_{kj}) \\
&\hspace{0.5cm}= -\frac{\gamma}{2\mu N}\sum_{k=1}^N(\kappa_{ik}-\kappa_{jk})^2-\frac{1}{2N}\sum_{k=1}^N(\kappa_{ik}-\kappa_{jk})(h_{ik}+h_{ki}-h_{jk}-h_{kj}) \\
&\hspace{0.5cm}=  -\frac{\gamma}{2\mu N}\sum_{k=1}^N(\kappa_{ik}-\kappa_{jk})^2+\frac{1}{N}\sum_{k=1}^N(\kappa_{ik}-\kappa_{jk})(R_{jk}-R_{ik}).
\end{aligned}
\end{align}
In \eqref{D-1-3}, we combine \eqref{D-1-4} and \eqref{D-1-5} to find the desired result.
\end{proof}
\end{lemma}

\begin{lemma}\label{L4.3} Let $(Z, K)$ be a solution to \eqref{D-1-2} with the initial  data $(Z^0, K^0)$ satisfying the following relations:
\[
\max_{i,j } \mathcal{L}_{ij}^0 < 1.
\]
Then, we have the following assertions:
\begin{enumerate}
\item
$R_{ij}$ and $\kappa_{ij}$ are strictly positive:
\[
R_{ij}(t)>0, \quad \kappa_{ij}(t)>0, \quad t \geq 0, \quad i,j\in \mathcal{N}.
\]
\item
$\mathcal{L}_{ij}$ is non-increasing function:
\[ \mathcal{L}_{ij}(t) \leq \mathcal{L}_{ij}^0, \quad t \geq 0, \quad i,j\in \mathcal{N}.   \]
\end{enumerate}
\end{lemma}
\begin{proof} Let $(i, j) \in {\mathcal N}^2 $ be fixed. Since $\kappa_{ij}^0>0$, by Lemma \ref{L2.1}, one has 
\[  \kappa_{ij}(t) > 0, \quad t > 0. \] 
Now, it follows from $\mathcal{L}_{ij}^0<1$ that 
\begin{equation} \label{D-1-6}
R_{ij}^0 > \frac{1}{4\mu N}\sum_{k=1}^{N}(\kappa_{ik}^0 -\kappa_{jk}^0)^2 \geq 0.
\end{equation}
We claim:
\[ R_{ij}(t)>0, \quad t \geq 0. \]
For this, we introduce a set ${\mathcal{T}}_{ij}$:
\[ {\mathcal T}_{ij} : = \{ \tau \in [0,\infty) : R_{ij}(t)>0, \quad t \in [0,\tau) \}. \]
Then, by \eqref{D-1-6} and continuity of $R_{ij}$, one has 
\[ {\mathcal T}_{ij} \not  = \emptyset. \]
Suppose that 
\[ t_{ij}^{*} := \sup  {\mathcal T}_{ij} < \infty. \]
Then, one has 
\[ \displaystyle\lim_{t \nearrow t^{*}_{ij}} R_{ij}(t) = 0. \]
We choose the index $(i_0,j_0)$ by
\[
(i_0,j_0) := \argmin_{(k,l)} t^{*}_{kl}.
\]
By the minimality of $t^{*}_{i_0j_0}$ and Lemma \ref{L4.2}, we have
\[
1-R_{i_0j_0}(t) \leq \mathcal{L}_{i_0j_0}(t) < \mathcal{L}_{i_0j_0}^0,\quad \text{so that} \quad R_{i_0j_0}(t) > 1-\mathcal{L}_{i_0j_0}^0>0, \quad t \in (0, t^{*}_{i_0j_0}).
\]
We take $t \nearrow t^{*}_{i_0j_0}$ to derive a contradiction:
\[
0=\lim_{t \nearrow t^{*}_{i_0j_0}} R_{i_0j_0}(t) \geq 1-\mathcal{L}_{i_0j_0}^0 > 0,
\]
Hence, we verified the claim:
\[  t^{*}_{i_0j_0}=\infty \quad \mbox{and} \quad R_{ij}(t)>0, \quad t \in (0, \infty). \]
By minimality, we have 
\[ t^{*}_{ij}=\infty \quad \mbox{for each index $(i,j)$}. \]
On the other hand, we have 
\[ 1-R_{ij} \geq 0. \]
Therefore, it follows from Lemma \ref{L4.2} that  the derivative of $\mathcal{L}_{ij}$ is not positive for every $t \in [0,\infty)$ which yields  the desired result.
\end{proof}
Before we provide a proof of Theorem \ref{T3.1}, we state Barbalat's lemma and Gr\"{o}nwall type lemma without proofs.
\begin{lemma} \label{L4.4}
\emph{(Barbalat's Lemma \cite{Ba})}
Suppose $f:[0,\infty) \to \bbr$ is uniformly continuous and satisfies
\[
\exists~~\lim_{t \to \infty} \int_0^t f(s)ds < \infty.
\]
Then, one has 
\[ \lim_{t \to \infty} f(t)=0. \]
\end{lemma}
\begin{lemma} \label{L4.5} \cite{HKLN}
Let $y: [0, \infty) \to [0, \infty)$ be a $\mathcal{C}^1$ function satisfying
\[
y' \leq -\alpha y + f, \quad t>0, \quad y(0)=y^0,
\]
where $\alpha$ is a positive constant and $f: [0, \infty) \to \bbr$ is a continuous function satisfying 
\[ \displaystyle\lim_{t \to \infty}f(t)=0. \]
Then $y$ satisfies
\[
y(t) \leq \frac{1}{\alpha} \Big( \max_{s \in [t/2,t]}|f(s)| \Big) + y^0 e^{-\alpha t}+\frac{\|f\|_{L^\infty}}{\alpha}e^{-\frac{\alpha t}{2}},\quad t\geq 0.
\]
\end{lemma}
\begin{proof} 
For a proof, we refer to Appendix A of \cite{HKLN}.
\end{proof}

Now we are ready to provide a proof of Theorem \ref{T3.1}. \newline

\noindent {\bf Proof of Theorem \ref{T3.1}}: Let $(Z, K)$ be a solution to \eqref{C-1-2} with the initial data $(Z^0, K^0)$ satisfying
 \[
 \max_{i,j} \mathcal{L}_{ij}^0<1, \quad \min_{i,j}\kappa_{ij}^0>0,
 \]
and we choose an index
\[
(I,J) := \argmax_{(i, j)\in \mathcal{N}^2} \mathcal{L}_{ij}^0.
\]
To apply Lemma \ref{L4.4} to $\kappa_{ij}$, we will show uniform boundedness of $\dot{\kappa}_{ij}$ in order to verify uniform continuity of $\kappa_{ij}$. Note that 
 \begin{equation}\label{D-2}
 \begin{aligned}
& \dot{\kappa}_{ij} = -\gamma\kappa_{ij} + \mu\|z_i-z_j\|^2 \leq - \gamma\kappa_{ij} + 4\mu,\\
 & \dot{\kappa}_{ij} = -\gamma\kappa_{ij} + \mu\|z_i-z_j\|^2 \geq - \gamma\kappa_{ij}
 \end{aligned}
 \end{equation}
 and from $\eqref{D-2}_1$ we obtain unfirom upper bound of $\kappa_{ij}$ as
 \begin{align}\label{D-2-0-0}
 \kappa_{ij}(t) \leq \kappa_{ij}^0e^{-\gamma t}+\frac{4\mu}{\gamma}(1-e^{-\gamma t}) \leq \kappa_{ij}^0 +\frac{4\mu}{\gamma} \leq \max_{k, l} \kappa_{kl}^0 + \frac{4\mu}{\gamma}.
 \end{align}
From Lemma \ref{L2.1}, $\kappa_{ij}$ is uniformly bounded below by 0. Therefore, $\kappa_{ij}$ is uniformly bounded. On the other hand, from \eqref{D-2} $\dot{\kappa}_{ij}$ is also uniformly bounded, therefore $\kappa_{ij}$ is uniformly continuous. \newline
 
 \noindent
It follows from Lemma \ref{L4.3} that
\[
1-R_{ij}(t) < \mathcal{L}^0_{ij}, \quad \text{therefore} \quad R_{ij}(t)>1-\mathcal{L}^0_{ij}>1-\mathcal{L}^0_{IJ}>0, \quad t>0.
\]
Therefore we have
\begin{align*}
 \dot{\mathcal{L}}_{ij} 
 &= -\left(\frac{1-R_{ij}}{N}\right)\sum_{k=1}^N(\kappa_{ik}R_{ik}+\kappa_{jk}R_{jk})-\frac{\gamma}{2\mu N}\sum_{k=1}^N(\kappa_{ik}-\kappa_{jk})^2 \quad   (\because \text{Lemma \ref{L4.2})}\\
&\leq -\left(\frac{1-R_{ij}}{N}\right)\sum_{k=1}^N\big((1-\mathcal{L}^0_{IJ})(\kappa_{ik}+\kappa_{jk})\big)-\frac{\gamma}{2\mu N}\sum_{k=1}^N(\kappa_{ik}-\kappa_{jk})^2  \quad (\because R_{ij}(t) > 1 - \mathcal{L}^0_{IJ})\\
 &\leq - \underbrace{\min \left\{ \frac{1}{N} \sum_{k=1}^N \big((1-\mathcal{L}^0_{IJ})(\kappa_{ik}+\kappa_{jk})\big),~2\gamma \right\}}_{=:\mathcal{K}_{ij}} \left( 1-R_{ij} + \frac{1}{4\mu N}\sum_{k=1}^N(\kappa_{ik}-\kappa_{jk})^2 \right) \\
  & = -\mathcal{K}_{ij}\mathcal{L}_{ij},
\end{align*}
for any index $i$ and $j$. This implies
\begin{equation}\label{D-2-0-1}
 \mathcal{L}_{ij}(t) \leq \mathcal{L}_{ij}^0 \exp \left( -\int_0^t \mathcal{K}_{ij}(s) ds \right).
\end{equation}
Since $\mathcal{K}_{ij} \geq 0$, we have only two possible cases:
\[ \mbox{either}~\int_0^{\infty} \mathcal{K}_{ij} dt = \infty \quad \mbox{or} \quad \int_0^{\infty} \mathcal{K}_{ij} dt < \infty. \]

\vspace{0.5cm}

\noindent $\bullet$~Case A $\left(\int_0^{\infty} \mathcal{K}_{ij} dt = \infty \right)$: We first assume that
\[
\int_0^{\infty} \mathcal{K}_{ij}(s) ds = \infty. 
\]
As $\mathcal{L}_{ij} \geq 0$, we use \eqref{D-2-0-1} to find 
\[ \displaystyle\lim_{t \to \infty} \mathcal{L}_{ij}(t) = 0. \]
In particular, we have 
\[ \displaystyle\lim_{t \to \infty} \| z_i(t) - z_j(t) \|=0. \]
Now we recall that dynamics of $\kappa_{ij}$ is defined by
\begin{equation}\label{D-2-0-3}
\dot{\kappa}_{ij} = -\gamma \kappa_{ij} + \mu \| z_i - z_j \|^2, \quad t>0.
\end{equation}
By  $\kappa_{ij}^0>0$, the assumption, Lemma \ref{L2.1}, one has 
\[  \kappa_{ij}(t)>0, \quad t \geq 0. \]
Therefore we can apply Lemma \ref{L4.5} to \eqref{D-2-0-3} to find
\[
\lim_{t \to \infty} \kappa_{ij}(t) = 0,
\]
verifying the desired result. \newline

\noindent $\bullet$~Case B $\left(\int_0^{\infty} \mathcal{K}_{ij} dt < \infty \right)$: Now assume that 
\[
\int_0^{\infty} \mathcal{K}_{ij} dt < \infty.
\]
and consider the set 
\[
A:= \left\{ t \in (0,\infty) ~:~
\frac{1}{N} \sum_{k=1}^N \big((1-\mathcal{L}^0_{IJ})(\kappa_{ik}+\kappa_{jk})\big) \geq 2\gamma
\right\}.
\]
Then for the Lebesgue measure $m$, we have following relation
\[
2\gamma m(A) = \int_A 2\gamma = \int_A \mathcal{K}_{ij} dt \leq \int_0^{\infty} \mathcal{K}_{ij} dt < \infty,
\]
therefore $m(A)<\infty$. This yields,
\begin{align}\label{D-2-0-4}
\begin{aligned}
&\int_0^{\infty} \frac{1}{N} \sum_{k=1}^N \big((1-\mathcal{L}^0_{IJ})(\kappa_{ik}+\kappa_{jk})\big) dt
\\ & \hspace{0.5cm} = \int_A \frac{1}{N} \sum_{k=1}^N \big((1-\mathcal{L}^0_{IJ})(\kappa_{ik}+\kappa_{jk})\big) dt
+ \int_{\bbr_+ \backslash A} \frac{1}{N} \sum_{k=1}^N \big((1-\mathcal{L}^0_{IJ})(\kappa_{ik}+\kappa_{jk})\big) dt
\\ & \hspace{0.5cm} = \int_A \frac{1}{N} \sum_{k=1}^N \big((1-\mathcal{L}^0_{IJ})(\kappa_{ik}+\kappa_{jk})\big) dt
+ \int_{\bbr_+ \backslash A} \mathcal{K}_{ij} dt \quad ( \because \text{Definition of  } A)
\\ & \hspace{0.5cm} \leq \int_A \frac{2}{N} \sum_{k=1}^N \left[\left(1-\mathcal{L}^0_{IJ}\right)\left(\max_{i,j} \kappa_{ij}^0 + \frac{4\mu}{\gamma}
\right)\right] dt
+ \int_{\bbr_+ \backslash A} \mathcal{K}_{ij} dt \quad (\because \eqref{D-2-0-0})
\\ &  \hspace{0.5cm}= 2 m(A) \left[\left(1-\mathcal{L}^0_{IJ}\right)\left(\max_{i,j} \kappa_{ij}^0 + \frac{4\mu}{\gamma}
\right)\right] dt
+ \int_{\bbr_+ \backslash A} \mathcal{K}_{ij} dt
\\ & \hspace{0.5cm} \leq 2 m(A) \left[\left(1-\mathcal{L}^0_{IJ}\right)\left(\max_{i,j} \kappa_{ij}^0 + \frac{4\mu}{\gamma}
\right)\right] dt
+ \int_0^\infty \mathcal{K}_{ij} dt < \infty.
\end{aligned}
\end{align}
Since $\mathcal{L}_{ij}^0<1$ from a priori condition and coupling gains are non-negative from Lemma \ref{L2.1}, \eqref{D-2-0-4} implies 
\[
\int_0^{\infty} (\kappa_{ik}+\kappa_{jk}) dt < \infty, \quad \forall i,j,k \in \mathcal{N}.
\]
Now we use the uniform continuity of $\kappa_{ij}$. By lemma \ref{L4.4}, we obtain
\begin{align}\label{D-2-0-5}
\lim_{t \to \infty} \kappa_{ik}(t) = \lim_{t \to \infty} \kappa_{jk} (t) = 0.
\end{align}
On the other hand, we recall the result of \eqref{D-1-4}:
\[
\dot{R}_{ij}=\frac{1}{2}(\dot{h}_{ij}+\dot{h}_{ji})= \frac{1}{N}\sum_{k=1}^N\Big((\kappa_{ik}-\kappa_{jk})(R_{jk}-R_{ik})+(R_{ik}\kappa_{ik}+R_{jk}\kappa_{jk})(1-R_{ij}) \Big).
\]
From uniform boundedness of $\kappa_{ij}$ and $R_{ij}$, we have uniform boundedness of $\dot{R}_{ij}$. Combining these all together with uniform boundedness of $\dot{\kappa}_{ij}$, we obtain uniform boundedness of
\[
 \ddot{\kappa}_{ij} =  -\gamma \dot{\kappa}_{ij} - 2 \mu \dot{R}_{ij},
\]
which leads to uniform continuity of $\dot{\kappa}_{ij}$. As integration of $\dot{\kappa}_{ij}$ is finite from \eqref{D-2-0-5}:
\[
\int_{0}^{\infty} \dot{\kappa}_{ij}(s)ds = -\kappa_{ij}^0,
\]
again from Lemma \ref{L4.4}, we can conclude that
\begin{align}\label{D-2-0-6}
\lim_{t \to \infty}\dot{\kappa}_{ik}(t) = \lim_{t \to \infty}\dot{\kappa}_{jk}(t) = 0.
\end{align}
Therefore, taking limit $t \to \infty$ to the dynamics
\[
\dot{\kappa}_{ij} = -\gamma \kappa_{ij} + \mu \| z_i - z_j \|^2,
\]
with \eqref{D-2-0-5} and \eqref{D-2-0-6} yield the desired result
\[
\lim_{t \to \infty}\|z_i(t) - z_j(t) \| = 0.
\]
\qed

\subsection{Hebbian coupling law}\label{sec:4.2}
Consider system \eqref{D-0-2} with the Hebbian coupling law $\eqref{D-0-3}_2$:
\begin{align*}
\begin{cases}
\displaystyle \dot{z}_j = \frac{1}{N}\displaystyle\sum_{k=1}^N \kappa_{jk} (z_k - R_{jk}z_j ), \quad t > 0,  \vspace{0.2cm} \\
\displaystyle \dot{\kappa}_{jk} = -\gamma \kappa_{jk} + \mu \left(1-\displaystyle\frac{\|z_j-z_k\|^2}{2}\right), \quad j, k \in {\mathcal N}, \vspace{0.2cm} \\
\displaystyle  (z_j, \kappa_{jk})(0) =(z_j^0, \kappa_{jk}^0) \in \bbh\bbs^{d} \times \bbr_+.
\end{cases}
\end{align*}
For a given system parameters $\gamma$ and $\mu$,  we choose positive constants $\kappa_m$ and $\kappa_M$ such that 
\begin{equation}\label{C-2-2}
\frac{1}{2}\kappa_M < \kappa_m \leq \frac{2\mu}{\gamma}\left(1-\frac{\kappa_m}{\kappa_M}\right).
\end{equation}
\begin{proposition} \label{P4.1}
For positive constants $\kappa_m$ and $\kappa_M$ satisfying \eqref{C-2-2}, suppose initial data satisfy   
\begin{equation}\label{C-2-3}
{\mathcal D}(Z^0) < \frac{2 \kappa_m}{\kappa_M} -1, \quad \kappa_m < \min_{i,j} \kappa^0_{ij},
\end{equation}
and let $(Z,K)$ be a solution to system \eqref{C-2-0} satisfying a priori assumption:
\begin{equation} \label{C-2-3-1}
\sup_{0 \leq t < \infty} \max_{i,j} \kappa_{ij}(t) \leq \kappa_M.
\end{equation}
Then,  there exist positive constants $D_0$ and $D_1$ such that
\[
{\mathcal D}(Z(t)) \leq D_0 e^{-D_1 t}, \quad t>0.
\]
\end{proposition}

\begin{proof}~By $\eqref{C-2-3}_2$ and continuity of solution, the set
\[ \mathcal{{\tilde T}}_{ij} := \{ \tau \in (0, \infty) ~ : ~ \kappa_{ij}(t) > \kappa_m, ~ \forall t \in (0, \tau) \} \not = \emptyset. \]
Now, we set
\[
\tilde{\mathcal{T}} := \bigcap_{i,j} \mathcal{{\tilde T}}_{ij} = \Big \{ \tau \in (0, \infty) ~ : ~ \min_{i,j}\kappa_{ij}(t) > \kappa_m, ~ \forall t \in (0,\tau) \Big \},
\quad {\tilde t}^{*} := \sup  \mathcal{{\tilde T}}.
\]
In the course of proof of Lemma \ref{L4.2}, we have
\begin{align*}
\dot{h}_{ij}+\dot{h}_{ji} =& \frac{2}{N}\sum_{k=1}^N\big((\kappa_{ik}-\kappa_{jk})(R_{jk}-R_{ik})+(R_{ik}\kappa_{ik}+R_{jk}\kappa_{jk})(1-R_{ij}) \big).
\end{align*}
In the sequel, for notational simplicity, we set 
\[  {\mathcal D}_{ij} := {\mathcal D}_{ij}(Z), \quad  {\mathcal D} := {\mathcal D}(Z), \quad {\mathcal D}^0 := {\mathcal D}(Z^0). \]
This and defining relation \eqref{C-2-3} of ${\mathcal D}$ imply
\begin{align*}
\begin{aligned}
\dot{{\mathcal D}}_{ij} &= \frac{1}{2}\frac{d}{dt}(2-h_{ij}-h_{ji})=-\frac{1}{2}(\dot{h}_{ij}+\dot{h}_{ji})\\
&= -\frac{1}{N} \sum_{k=1}^N \big( (\kappa_{ik}+\kappa_{jk}) {\mathcal D}_{ij}  \big)
+ \frac{1}{N} \sum_{k=1}^N  \big( \kappa_{ik} {\mathcal D}_{ik} + \kappa_{jk} {\mathcal D}_{jk} \big) {\mathcal D}_{ij} \\
&\hspace{0.2cm} + \frac{1}{N} \sum_{k=1}^N (\kappa_{ik} - \kappa_{jk} )({\mathcal D}_{jk}- {\mathcal D}_{ik}) \\
&=: {\mathcal J}_{11} + {\mathcal J}_{12}  + {\mathcal J}_{13}.
\end{aligned}
\end{align*}
From now on, we will regard $i$ and $j$ as a function of $t$.  For each $t$, we assume that $i,j$ are indices such that 
\[ {\mathcal D}= {\mathcal D}_{ij}. \]
It follows from \eqref{C-2-3} and the definition of ${\tilde t}^{*}$ that
\begin{equation}\label{D-2-1}
{ \mathcal{J}_{11} \geq 2\kappa_m {\mathcal D} },  \qquad  \mathcal{J}_{12} \leq 2\kappa_M {\mathcal D}^2,
 \qquad 
 \mathcal{J}_{13} \leq 2(\kappa_M - \kappa_m) {\mathcal D},
 \quad t \in [0, {\tilde t}^{*}).
\end{equation}
This leads to a differential inequality:
\begin{equation}\label{D-2-2}
\dot{{\mathcal D}} < -2(2\kappa_m-\kappa_M) {\mathcal D} +2\kappa_M {\mathcal D}^2, \quad  t \in [0, {\tilde t}^{*}),
\end{equation}
where the factor $(2\kappa_m-\kappa_M)$ is positive from the first inequality of \eqref{C-2-3}. \newline

We apply the comparison principle to \eqref{D-2-2} to get 
\[
{\mathcal D}(t) \leq \frac{1}{\left(\frac{1}{{\mathcal D}^0}-\frac{2\kappa_m - \kappa_M}{\kappa_M}\right)e^{(2\kappa_m-\kappa_M)t}
+\frac{2\kappa_m-\kappa_M}{\kappa_M}}, \quad  t \in [0,{\tilde t}^{*}),
\]
i.e. exponential decay occurs in $t \in [0, {\tilde t}^{*})$. Hence the proof is done if we verify ${\tilde t}^{*}=\infty$.  \newline

Now we use the initial condition ${\mathcal D}^0 < \frac{2 \kappa_m}{\kappa_M} -1$ and \eqref{D-2-2} to obtain
\[
-2(2\kappa_m-\kappa_M) {\mathcal D} +2\kappa_M {\mathcal D}^2 < 0, \quad \text{whenever} \quad {\mathcal D} \in \left(0, \frac{2\kappa_m}{\kappa_M}-1\right).
\]
Therefore, ${\mathcal D}$ is decreasing in $ t \in [0, {\tilde t}^{*})$. Hence we have 
\[
\dot\kappa_{ij} = -\gamma\kappa_{ij} + \mu - \mu {\mathcal D}_{ij} \geq -\gamma\kappa_{ij} + \mu - \mu {\mathcal D} \geq -\gamma \kappa_{ij} + 2\mu \frac{\kappa_M-\kappa_m}{\kappa_M}, \quad  t \in [0, {\tilde t}^{*}).
\]
By comparison principle, one has
\begin{align*}
\kappa_{ij} \geq \left( \kappa_{ij}^0 - \frac{2\mu}{\gamma}\frac{\kappa_M-\kappa_m}{\kappa_M}\right)e^{-\gamma t} + \frac{2\mu}{\gamma}\left(\frac{\kappa_M-\kappa_m}{\kappa_M}\right)
 \geq \left( \kappa_{ij}^0 - \kappa_m\right)e^{-\gamma t} + \kappa_m, \quad  t \in [0, {\tilde t}^{*}),
 \end{align*}
where the last equality holds from the second inequality of \eqref{C-2-2}. By definition of $\tilde{\mathcal{T}}$, there exist indices $k$ and $l$ such that 
\[ \kappa_m = \lim_{t \nearrow {\tilde t}^{*}} \kappa_{kl}. \]
Therefore if ${\tilde t}^{*}$ is finite, one has
\begin{align*}
\kappa_m = \lim_{t \nearrow \tilde{t}^{*}} \kappa_{kl}
\geq \lim_{t \nearrow \tilde{t}^*} \left( \kappa_{kl}^0 - \kappa_m\right)e^{-\gamma t} + \kappa_m
= \left( \kappa_{kl}^0 - \kappa_m\right)e^{-\gamma \tilde{t}^*} + \kappa_m > \kappa_m,
 \end{align*}
which is contradictory, and we obtain our desired result.
 \end{proof}
 
 \vspace{0.5cm}
 
 Now we are ready to provide a proof of our second main result. \newline
\begin{proof}[{\bf Proof of Theorem \ref{T3.2}}] Recall the conditions \eqref{C-2-4}: 
\begin{equation} \label{D-3}
0 < \kappa < \min \left\{  \frac{\mu}{\gamma},~\min_{i,j}\kappa_{ij}^0 \right \}, \quad \max \left\{  \max_{i,j} \kappa^0_{ij}, ~\frac{\mu}{\gamma} \right \} \leq \frac{2\mu \kappa}{2\mu-\gamma \kappa}, 
\quad {\mathcal D}(Z^0) < 1-\frac{\gamma}{\mu} \kappa.
\end{equation}
Now, it suffices to show that the above conditions satisfy \eqref{C-2-3} and \eqref{C-2-3-1}:
\begin{equation} \label{D-4}
\kappa_m < \min_{i,j} \kappa^0_{ij}, \quad {\mathcal D}(Z^0) < \frac{2 \kappa_m}{\kappa_M} -1, \quad \quad \sup_{0 \leq t < \infty} \max_{i,j} \kappa_{ij}(t) \leq \kappa_M.
\end{equation}
We first figure out $\kappa_m$ and $\kappa_M$ satisfying \eqref{C-2-2}:
\begin{equation}\label{D-4-1}
\frac{1}{2}\kappa_M < \kappa_m, \quad \kappa_m \leq \frac{2\mu}{\gamma}\left(1-\frac{\kappa_m}{\kappa_M}\right).
\end{equation}
Since $\kappa$ is a candidate of $\kappa_m$, we assume that $\kappa_m$ satisfies \eqref{D-3}. Rewriting $\eqref{D-4-1}_2$, we have 
\begin{align}\label{D-4-2}
 \kappa_m \leq \frac{2\mu}{\gamma}\left(1-\frac{\kappa_m}{\kappa_M}\right)
 \Longleftrightarrow
  \frac{2\mu \kappa_m}{2\mu-\gamma \kappa_m} \leq \kappa_M.
\end{align}
Optimizing $\kappa_M$ under \eqref{D-4-2}, we have
\[
\frac{2\mu \kappa_m}{2\mu-\gamma \kappa_m} = \kappa_M.
\]
Therefore, as we set
\begin{equation} \label{D-5}
\kappa_m= \kappa \quad \mbox{and} \quad \kappa_M= \frac{2\mu \kappa}{2\mu-\gamma \kappa},
\end{equation}
$\eqref{D-4-1}_2$ is achieved. In particular, as $\kappa$ satisfies $\eqref{D-3}_1$, we have
\[
\frac{1}{2}\kappa_M < \kappa_m
\Longleftrightarrow
\kappa < \frac{\mu}{\gamma},
\]
which is true from $\eqref{D-3}_1$. Hence $\eqref{D-4-1}_1$ is achieved. \vspace{0.2cm}\\
\noindent $\bullet$~(Verification of \eqref{C-2-3}):~Clearly, $\eqref{D-3}_1$ implies $\eqref{D-4}_1$. By the setting \eqref{D-5}, one has 
\[  \frac{2 \kappa_m}{\kappa_M} -1 = \frac{2\kappa}{\frac{2\mu \kappa}{2\mu - \gamma k}} - 1 = 1- \frac{\gamma \kappa}{\mu}. \]
Hence $\eqref{D-3}_3$ is equivalent to $\eqref{D-4}_2$ under the setting \eqref{D-5}. 

\vspace{0.2cm}

\noindent $\bullet$~(Verification of \eqref{C-2-3-1}):~Note that 
\begin{align*}
\begin{aligned}
\kappa_{ij}(t) &= e^{-\gamma t}\left(\kappa_{ij}^0 + \int_0^t \mu e^{\gamma s}R_{ij} ds \right) 
\leq  e^{-\gamma t}\left(\kappa^0_{ij} + \int_0^t \mu e^{\gamma s} ds \right)  \\
&\leq \left(\kappa^0_{ij} -\frac{\mu}{\gamma} \right)e^{-\gamma t} + \frac{\mu}{\gamma} \leq \max \left\{ \max_{i,j} \kappa^0_{ij}, ~\frac{\mu}{\gamma} \right\} \leq \kappa_M.
\end{aligned}
\end{align*}
Finally, we can apply Proposition \ref{P4.1} to derive the desired estimate. 
\end{proof}

\section{Collective dynamics under asymptotic SL coupling gain pair}\label{sec:5}
\setcounter{equation}{0}
In this section, we study the emergent dynamics of the system \eqref{A-0} for a general coupling gain pair $(\kappa_{ij}, \lambda_{ij})$:
\[ \kappa_{ij}(t) > 0, \quad \lambda_{ij}(t) \in \bbr, \quad t \geq 0, \quad i, j \in {\mathcal N}. \]

\vspace{0.2cm}

Recall that our governing system is given as follows:
\begin{equation}\label{E-0}
\begin{cases}
\displaystyle \dot{z}_j= \frac{1}{N}\sum_{k=1}^N\kappa_{jk}( z_k-R_{jk} z_j)+\frac{1}{N}\sum_{k=1}^N\tilde{\lambda}_{jk}(\langle z_j, z_k\rangle-\langle z_k, z_j\rangle)z_j, \quad t > 0, \vspace{0.2cm} \\
\displaystyle \dot{\kappa}_{ij}=-\gamma\kappa_{ij}+\mu\Gamma_0(z_i,z_j),\quad \dot{\tilde{\lambda}}_{ij}=-\gamma\tilde{\lambda}_{ij}+\mu\tilde{\Gamma}(z_i, z_j), \vspace{0.2cm} \\
\displaystyle (z_j, \kappa_{ij}, {\tilde \lambda}_{ij})(0) =(z_j^0, \kappa_{ij}^0, {\tilde \lambda}_{ij}^0) \in \bbh\bbs^{d} \times \bbr_+ \times \bbr, \quad i, j\in \mathcal{N}.
\end{cases}
\end{equation}
Before we proceed analysis on the emergent behavior of \eqref{E-0}, we observe that  the ratio of $\kappa_{ij}$ and $\tilde{\lambda}_{ij}$ is bounded below by the ratio of $\Gamma_0$ and $\tilde{\Gamma}$ in the following lemma. 
\begin{lemma}\label{L5.1}
Suppose that coupling gains and coupling law satisfy
\[
c\kappa^0_{ij} \geq \tilde{\lambda}^0_{ij} \quad \text{and} \quad c \Gamma_0(z_i,z_j) \geq \tilde{\Gamma}(z_i,z_j), \quad i, j \in {\mathcal N}
\]
for some constant $c>0$, and let $(Z, K, \tilde{\Lambda})$ be a solution to \eqref{E-1-7}. Then, one has 
\[
c \kappa_{ij}(t) \geq \tilde{\lambda}_{ij}(t), \quad \forall t>0. 
\]
\end{lemma}
\begin{proof} We use \eqref{E-1-7} to see
\[
\frac{d}{dt}\big(c \kappa_{ij} - \tilde{\lambda}_{ij} \big) = -\gamma\big(c \kappa_{ij} - \tilde{\lambda}_{ij} \big)+\mu\big( c \Gamma_0(z_i,z_j) - \tilde{\Gamma}(z_i,z_j) \big) \geq -\gamma\big(c \kappa_{ij} - \tilde{\lambda}_{ij} \big).
\]
Therefore, we have
\[
 c \kappa_{ij}(t)-\tilde{\lambda}_{ij}(t) \geq e^{-\gamma t}\big(c\kappa^0_{ij} -\tilde{\lambda}^0_{ij} \big) \geq 0.
\]
and this is the desired result.
\end{proof}
Parallel to the presentation in Section \ref{sec:4}, in what follows, we consider two type of coupling laws for $\Gamma_0$ as in Section \ref{sec:4}:
\[ \Gamma_0(z, {\tilde z}): \|z - {\tilde z} \|^2, \quad 1-\frac{1}{2} \|z - {\tilde z} \|^2.  \]
\subsection{Anti-Hebbian coupling law} \label{sec:5.1}
In this subsection, we study emergent dynamics of \eqref{E-1-7} with the anti-Hebbian coupling law:  
\begin{equation}\label{E-2}
\begin{cases}
\displaystyle \dot{z}_j= \displaystyle\frac{1}{N}\sum_{k=1}^N\kappa_{jk} \Big( z_k-R_{jk} z_j \Big) +\frac{1}{N}\sum_{k=1}^N\tilde{\lambda}_{jk} \Big(\langle z_j, z_k\rangle-\langle z_k, z_j\rangle \Big)z_j,  \quad t > 0, \vspace{0.2cm} \\
\displaystyle \dot{\kappa}_{jk}=-\gamma\kappa_{jk}+\mu\|z_j-z_k\|^2,\quad \dot{\tilde{\lambda}}_{jk}=-\gamma\tilde{\lambda}_{jk}+\mu\tilde{\Gamma}(z_j, z_k), \vspace{0.2cm} \\
\displaystyle (z_j, \kappa_{jk}, {\tilde \lambda}_{jk})(0) =(z_j^0, \kappa_{jk}^0, {\tilde \lambda}_{jk}^0) \in \bbh\bbs^d \times \bbr_+ \times \bbr, \quad  j, k\in \mathcal{N}.
\end{cases}
\end{equation}
As in Section \ref{sec:4}, we study the temporal evolution of the Lyapunov functional ${\mathcal L}_{ij}$ introduced in \eqref{C-1-3} . 
\begin{lemma}\label{L5.2}
Let $(Z, K, {\tilde \Lambda})$ be a solution to \eqref{E-2}. Then, the functional ${\mathcal L}_{ij}$ satisfies 
\begin{align}
\begin{aligned} \label{E-2-0}
\frac{d}{dt} {\mathcal{L}}_{ij} &= -\left(\frac{1-R_{ij}}{N}\right)\sum_{k=1}^N(\kappa_{ik}R_{ik}+\kappa_{jk}R_{jk})-\frac{\gamma}{2\mu N}\sum_{k=1}^N(\kappa_{ik}-\kappa_{jk})^2 \\
&\hspace{0.4cm} -\frac{2 I_{ij}}{N}\sum_{k=1}^N(\tilde{\lambda}_{ik}I_{ik}-\tilde{\lambda}_{jk}I_{jk}).
\end{aligned}
\end{align}
\end{lemma}
\begin{proof} By definition of ${\mathcal L}_{ij}$, one has
\begin{equation} \label{E-2-1}
\dot{\mathcal{L}}_{ij} = -\frac{1}{2}(\dot{h}_{ij}+\dot{h}_{ji}) + \frac{1}{2\mu N}\sum_{k=1}^N (\kappa_{ik}-\kappa_{jk})(\dot{\kappa}_{ik}-\dot{\kappa}_{jk}).
\end{equation}
Next, we estimate two terms in the R.H.S. of \eqref{E-2-1} separately. \newline

\noindent $\bullet$~(Estimate of the first term in \eqref{E-2-1}): By straightforward calculations, one has 
\begin{align}
\begin{aligned} \label{E-2-2}
&\dot{h}_{ij}+\dot{h}_{ji} = \langle z_i, \dot{z}_j \rangle + \langle \dot{z}_i, z_j \rangle + \langle z_j, \dot{z}_i \rangle + \langle \dot{z}_j, z_i \rangle = \frac{2}{N}\sum_{k=1}^N\mathrm{Re}( \langle z_i, \dot{z}_j \rangle + \langle \dot{z}_i, z_j \rangle )\\
&= \frac{2}{N}\sum_{k=1}^N\mathrm{Re}\Big[  \kappa_{jk}(h_{ik}-R_{jk}h_{ij}  )+\kappa_{ik}(h_{kj}-R_{ik}h_{ij}  )
+ \tilde{\lambda}_{jk}(h_{jk}-h_{kj})h_{ij} + \tilde{\lambda}_{ik}(h_{ki}-h_{ik})h_{ij}) \Big]\\
& = \frac{2}{N}\sum_{k=1}^N \Big[  \kappa_{jk}(R_{ik}-R_{jk}R_{ij}  )+\kappa_{ik}(R_{jk}-R_{ik}R_{ji}  )\big)
-2 \tilde{\lambda}_{jk} I_{jk}I_{ij} -2 \tilde{\lambda}_{ik}I_{ki}I_{ij}) \Big].
\end{aligned}
\end{align}

\vspace{0.2cm}

\noindent $\bullet$~(Estimate of the second term in \eqref{E-2-1}): Again, one has 
\begin{align}
\begin{aligned} \label{E-2-3}
&\frac{1}{2\mu N}\sum_{k=1}^N(\kappa_{ik}-\kappa_{jk})(\dot{\kappa}_{ik}-\dot{\kappa}_{jk}) \\
& \hspace{0.5cm} = \frac{1}{2\mu N}\sum_{k=1}^N (\kappa_{ik}-\kappa_{jk})(-\gamma\kappa_{ik}-\mu h_{ik}-\mu h_{ki} +\gamma\kappa_{jk}+\mu h_{jk} + \mu h_{kj}) \\
& \hspace{0.5cm}  = -\frac{\gamma}{2\mu N}\sum_{k=1}^N(\kappa_{ik}-\kappa_{jk})^2-\frac{1}{2N}\sum_{k=1}^N(\kappa_{ik}-\kappa_{jk})(h_{ik}+h_{ki}-h_{jk}-h_{kj}) \\
& \hspace{0.5cm}=  -\frac{\gamma}{2\mu N}\sum_{k=1}^N(\kappa_{ik}-\kappa_{jk})^2-\frac{1}{N}\sum_{k=1}^N(\kappa_{ik}-\kappa_{jk})(R_{ik}-R_{jk}).
\end{aligned}
\end{align}
In \eqref{E-2-1}, we combine \eqref{E-2-2} and \eqref{E-2-3} to obtain the desired estimate. 
\end{proof}

\begin{lemma}\label{L5.3}
Suppose that the following relations hold:
\begin{align}\label{E-2-4}
\tilde{\lambda}_{ij}^0={\tilde \lambda}^0, \quad ~i, j \in \mathcal{N} \quad \text{and} \quad \tilde{\Gamma}(t) \equiv 0, \quad\forall~  t>0
\end{align}
for some constant $\tilde\lambda^0$, and let $(Z, K, {\tilde \Lambda})$ be a solution to \eqref{E-2}. Then, the following assertions hold:
\begin{enumerate}
\item
There exists a function $\tilde\lambda = {\tilde \lambda}(\cdot)$ such that
\[
\tilde{\lambda}_{ij}(t) = {\tilde \lambda} (t), \quad t > 0, \quad  \forall~ i,j \in \mathcal{N}.
\]
\item
The functional ${\mathcal L}_{ij}$ satisfies
\[
 \dot{\mathcal{L}}_{ij} \leq -\left(\frac{1-R_{ij}}{N}\right)\sum_{k=1}^N \Big(\kappa_{ik}R_{ik}+\kappa_{jk}R_{jk}-4|\tilde \lambda| \Big)-\frac{\gamma}{2\mu N}\sum_{k=1}^N(\kappa_{ik}-\kappa_{jk})^2, \quad t > 0.
\]
\end{enumerate}
\end{lemma}
\begin{proof} It follows from \eqref{E-2} and \eqref{E-2-4} that 
\begin{equation}\label{E-2-5}
\begin{cases}
\displaystyle \dot{z}_j= \displaystyle\frac{1}{N}\sum_{k=1}^N\kappa_{jk} \Big( z_k-R_{jk} z_j \Big) +\frac{1}{N}\sum_{k=1}^N\tilde{\lambda}_{jk} \Big(\langle z_j, z_k\rangle-\langle z_k, z_j\rangle \Big)z_j, \quad t > 0, \vspace{0.2cm} \\
\displaystyle \dot{\kappa}_{jk}=-\gamma\kappa_{jk}+\mu\|z_j-z_k\|^2,\quad \dot{\tilde{\lambda}}_{jk} =-\gamma\tilde{\lambda}_{jk}, \vspace{0.2cm} \\
\displaystyle (z_j, \kappa_{jk}, {\tilde \lambda}_{jk})(0) =(z_j^0, \kappa_{jk}^0, {\tilde \lambda}_{jk}^0) \in \bbh\bbs^d \times \bbr_+ \times \bbr, \quad j, k \in \mathcal{N}.
\end{cases}
\end{equation}

\vspace{0.2cm}

\noindent (i)~By $\eqref{E-2-5}_2$, one has 
\[
\tilde{\lambda}_{ij}(t)={\tilde \lambda}^0 e^{-\gamma t} =: {\tilde \lambda}(t), \quad t>0,
\]
which yields the first assertion. 

\vspace{0.2cm}

 \noindent (ii)~Now, we estimate the last term of \eqref{E-2-0} as follows.
\[
\frac{2 I_{ij}}{N}\sum_{k=1}^N(\tilde{\lambda}_{ik}I_{ik}-\tilde{\lambda}_{jk}I_{jk}) =\frac{2 I_{ij}}{N}\sum_{k=1}^N \big( \tilde{\lambda}_{ik}(I_{ik}-I_{jk})-I_{jk}(\tilde{\lambda}_{jk}-\tilde{\lambda}_{ik}) \big) =\frac{2 I_{ij}}{N}\sum_{k=1}^N \tilde{\lambda} (I_{ik}-I_{jk}).
\]
We use the triangle inequality and the Cauchy-Schwarz inequality to find
\begin{align}
\begin{aligned} \label{E-2-6}
& | I_{ij} | = \sqrt{1-R_{ij}^2} = \sqrt{(1-R_{ij})(1+R_{ij})} \leq \sqrt{2(1-R_{ij})}, \\
& | I_{ik} - I_{jk} | = \left| \frac{1}{2i} (h_{ik}-h_{ki}-h_{jk}+h_{kj}) \right| = \frac{1}{2} \left| \langle z_i-z_j, z_k \rangle + \langle z_k , z_j - z_i \rangle \right| \\
& \hspace{1.7cm} \leq \| z_i - z_j \| = \sqrt{2(1-R_{ij})}.
\end{aligned}
\end{align}
Finally, we combine \eqref{E-2-0} and \eqref{E-2-6} to obtain
\begin{align*}
\begin{aligned}
 \dot{\mathcal{L}}_{ij} &= -\left(\frac{1-R_{ij}}{N}\right)\sum_{k=1}^N(\kappa_{ik}R_{ik}+\kappa_{jk}R_{jk})-\frac{\gamma}{2\mu N}\sum_{k=1}^N(\kappa_{ik}-\kappa_{jk})^2-\frac{2 I_{ij}}{N}\sum_{k=1}^N\tilde{\lambda}_{ik}(I_{ik}-I_{jk})\\
&\leq -\left(\frac{1-R_{ij}}{N}\right)\sum_{k=1}^N(\kappa_{ik}R_{ik}+\kappa_{jk}R_{jk})-\frac{\gamma}{2\mu N}\sum_{k=1}^N(\kappa_{ik}-\kappa_{jk})^2+
\left|\frac{2 I_{ij}}{N}\sum_{k=1}^N\tilde{\lambda}_{ik}(I_{ik}-I_{jk})\right| \\
&\leq -\left(\frac{1-R_{ij}}{N}\right)\sum_{k=1}^N(\kappa_{ik}R_{ik}+\kappa_{jk}R_{jk})-\frac{\gamma}{2\mu N}\sum_{k=1}^N(\kappa_{ik}-\kappa_{jk})^2
+ 4\left(\frac{1-R_{ij}}{N}\right) \sum_{k=1}^N | \tilde \lambda | \\ 
& = -\left(\frac{1-R_{ij}}{N}\right)\sum_{k=1}^N \Big(\kappa_{ik}R_{ik}+\kappa_{jk}R_{jk}-4 |\tilde \lambda| \Big)-\frac{\gamma}{2\mu N}\sum_{k=1}^N(\kappa_{ik}-\kappa_{jk})^2,
\end{aligned}
\end{align*}
and this is the desired result.
\end{proof}
\begin{lemma}\label{L5.4} 
Suppose that the following relations
\[
\tilde{\lambda}_{ij}^0={\tilde \lambda}^0, \quad i,j \in \mathcal{N} \quad \text{and} \quad \tilde{\Gamma}(t) =  0, \quad t>0
\]
hold for some constant ${\tilde \lambda}^0$, and let $(Z, K, {\tilde \Lambda})$ be a solution to \eqref{E-2} with the initial data satisfying
\[
\quad 1 > \max_{i,j}\frac{2|\tilde\lambda^0|}{\kappa_{ij}^0}+\max_{k,l}\mathcal{L}_{kl}^0.
\]
Then, $\mathcal{L}_{ij}$ is non-increasing.
\end{lemma}
\begin{proof}
\noindent If $\tilde \lambda^0=0$, by the same argument in a proof of Lemma \ref{L4.3}, we are done. We choose a constant $c$ satisfying 
\begin{align}\label{X-3}
c\kappa_{ij}^0 \geq |\tilde\lambda^0|,
\end{align}
for any indices $i$ and $j$. Since $\tilde{\Gamma} \equiv 0$, by Lemma \ref{L5.1} we have
\[
c \kappa_{ij}(t) \geq |\tilde{\lambda}(t)|, \quad t>0.
\]
Hence, one has
\begin{align}
\begin{aligned}\label{X-3-0}
 \dot{\mathcal{L}}_{ij} 
&\leq -\left(\frac{1-R_{ij}}{N}\right)\sum_{k=1}^N(\kappa_{ik}R_{ik}+\kappa_{jk}R_{jk}-4|\tilde\lambda|)-\frac{\gamma}{2\mu N}\sum_{k=1}^N(\kappa_{ik}-\kappa_{jk})^2\\
&\leq -\left(\frac{1-R_{ij}}{N}\right)\sum_{k=1}^N\big(\kappa_{ik}(R_{ik}-2c)+\kappa_{jk}(R_{jk}-2c)\big)-\frac{\gamma}{2\mu N}\sum_{k=1}^N(\kappa_{ik}-\kappa_{jk})^2.
\end{aligned}
\end{align}
Now we recall that $\kappa_{ij}$ is positive by Lemma \ref{L2.1} because $\kappa_{ij}^0>0$. Therefore if we assume
\begin{align}\label{X-4}
\min_{k,l} R_{kl}^0>2c,
\end{align}
then
\[
\tilde{\mathcal{T}}^*_{ij} := \{ \tau \in (0,\infty) : R_{ij}(t)>2c , \quad t \in (0,\tau) \}, \quad
\]
is non-empty for each $i,j$ and we can introduce
\[
\tilde{t}^*_{ij} := \sup \tilde{\mathcal{T}}^*_{ij}, \quad
(i_1,j_1) := \argmin_{k, l} {\tilde t}^*_{kl}.
\]
Therefore, it follows from the minimality of $\tilde{t}^*_{i_1j_1}$ that
\[ \mathcal{L}_{i_1j_1}(t) < \mathcal{L}_{i_1j_1}^0, \quad t \in (0, {\tilde t}^*_{i_1j_1}).\]
By definition \eqref{D-1-0-0}, one has 
\[ 1-R_{i_1j_1}(t) \leq \mathcal{L}_{i_1j_1}(t), \quad t \geq 0. \]
Hence we have
\[
1-R_{i_1j_1}(t) \leq \mathcal{L}_{i_1j_1}(t) < \mathcal{L}_{i_1j_1}^0,\quad \text{so that} \quad R_{i_1j_1}(t) > 1-\mathcal{L}_{i_1j_1}^0 >0 \quad \text{for} \quad t \in (0, {\tilde t}^*_{i_1j_1}).
\]
Therefore by taking $t \nearrow {\tilde t}^*_{i_1j_1}$, from the continuity of $R_{i_1j_1}$we obtain
\begin{align*}
R_{i_1j_1}({\tilde t}^*_{i_1j_1}) \geq 1-\mathcal{L}_{i_1j_1}^0 >0.
\end{align*}
If we impose the relationship
\begin{align}\label{X-5}
1-\mathcal{L}_{ij}^0 > 2c,
\end{align}
one can obtain
\[
\tilde{t}^*_{i_1j_1}<\infty \Longrightarrow
1 - \mathcal{L}_{i_1j_1}^0 > 2c = R_{i_1j_1}({\tilde t}^*_{i_1j_1}) \geq 1-\mathcal{L}_{i_1j_1}^0,
\]
which is contradictory. Therefore 
\begin{align}\label{X-5-0}
\tilde{t}^*_{i_1j_1}=\infty,
\end{align}
so that $\tilde{t}^*_{ij}=\infty$ and $\mathcal{L}_{ij}$ is a non-increasing function of $t$. 
Now we choose optimal $c$ satisfying \eqref{X-3}. Namely, we set
\[
c=\max_{i,j}\frac{|\tilde{\lambda}^0|}{\kappa^0_{ij}},
\]
and \eqref{X-4} and \eqref{X-5} can be specified as
\begin{align}\label{X-6}
\min_{k,l} R_{kl}^0>2\max_{i,j}\frac{|\tilde{\lambda}^0|}{\kappa^0_{ij}}
 \quad \text{and} \quad
1-\mathcal{L}_{ij}^0 > 2\max_{i,j}\frac{|\tilde{\lambda}^0|}{\kappa^0_{ij}},
\end{align}
respectively. Since $R_{ij}^0 > 1 - \mathcal{L}_{ij}^0$, \eqref{X-6} is achieved from a priori condition and we have a desired result.
\end{proof}

\vspace{0.5cm}
As a consequence of continuity argument used in a proof of Lemma \ref{L5.4}, we obtain the following result.
\begin{lemma}\label{L5.5}
Suppose that the following relations
\[
\tilde{\lambda}_{ij}^0={\tilde \lambda}^0, \quad i,j \in \mathcal{N} \quad \text{and} \quad \tilde{\Gamma}(t) =  0, \quad t>0
\]
hold for some constant ${\tilde \lambda}^0$, and let $(Z, K, {\tilde \Lambda})$ be a solution to \eqref{E-2} with the initial data satisfying
\begin{align}\label{ZZ-1}
\quad 1 > \max_{i,j}\frac{2|\tilde\lambda^0|}{\kappa_{ij}^0}+\max_{k,l}\mathcal{L}_{kl}^0.
\end{align}
Then we have
\[
\min_{K,L}R_{KL}(t)>2\max_{I,J}\frac{|\tilde{\lambda}^0|}{\kappa_{IJ}^0}, \quad t>0.
\]
In particular, for any index $k$, we have
\[
\kappa_{ik}\left(R_{ik}-\max_{l,m}\frac{2|\tilde\lambda^0|}{\kappa_{lm}^0}\right) > 0.
\]
\begin{proof}
This is a direct consequence of \eqref{X-5-0}.
\end{proof}
\end{lemma}

\vspace{0.5cm}

Now we are ready to provide a proof of Theorem \ref{T3.3}. \newline

\begin{proof}[{\bf Proof of Theorem \ref{T3.3}}] Since $\max_{i,j}\frac{|\tilde\lambda^0|}{\kappa_{ij}^0}$ satisfies relationship \eqref{X-3}, from \eqref{X-3-0} it follows that 
\begin{align*}
 \dot{\mathcal{L}}_{ij} &\leq -\left(1-R_{ij}\right)
 \underbrace{\frac{1}{N}\sum_{k=1}^N\left(\kappa_{ik}\left(R_{ik}-\max_{i,j}\frac{2|\tilde\lambda^0|}{\kappa_{ij}^0}\right)+\kappa_{jk}\left(R_{jk}-\max_{i,j}\frac{2|\tilde\lambda^0|}{\kappa_{ij}^0}\right)\right)}
 _{=:\mathcal{O}_{ij}} -\frac{\gamma}{2\mu N}\sum_{k=1}^N(\kappa_{ik}-\kappa_{jk})^2\\
 &\leq - \underbrace{\min \left\{ \mathcal{O}_{ij},~2\gamma \right\}}_{=:\mathcal{M}_{ij}} \left( 1-R_{ij} + \frac{1}{4\mu N}\sum_{k=1}^N(\kappa_{ik}-\kappa_{jk})^2 \right)  =: -\mathcal{M}_{ij}\mathcal{L}_{ij}.
\end{align*}
This leads to
\begin{equation}\label{E-3}
 \mathcal{L}_{ij}(t) \leq \mathcal{L}_{ij}^0 \exp \left( -\int_0^t \mathcal{M}_{ij}(s) ds \right).
\end{equation}
Since $\mathcal{O}_{ij}$ is positive from Lemma \ref{L5.5}, so is $\mathcal{M}_{ij}$. Hence we have only two possible cases:
\[ \mbox{either}~\int_0^{\infty} \mathcal{M}_{ij} dt = \infty \quad \mbox{or} \quad \int_0^{\infty} \mathcal{M}_{ij} dt < \infty. \]

\vspace{0.5cm}

\noindent $\bullet$ Case A ($\int_0^{\infty} \mathcal{M}_{ij} dt = \infty$): In this case, as we did in a proof of Theorem \ref{T3.1}, we have
\[
\lim_{t \to \infty} \mathcal{L}_{ij}(t) = 0 
\Longrightarrow \lim_{t \to \infty} \| z_i(t) - z_j(t) \|=0.
\]
Therefore, it follows from Lemma \ref{L4.5} that
\[
\dot{\kappa}_{ij} = -\gamma \kappa_{ij} + \mu \| z_i - z_j \|^2
\Longrightarrow
\lim_{t \to \infty} \kappa_{ij}(t) = 0.
\]

\noindent $\bullet$ Case B ($\int_0^{\infty} \mathcal{M}_{ij} dt < \infty$):  consider the set 
\[
B:= \left\{ t \in (0,\infty) \mid
\mathcal{O}_{ij} \geq 2\gamma
\right\}.
\]
Then for the Lebesgue measure $m$, we have
\[
2\gamma m(B) = \int_B 2\gamma = \int_B \mathcal{M}_{ij} dt \leq \int_0^{\infty} \mathcal{M}_{ij} dt < \infty, 
\]
therefore $m(B)<\infty$. This yields,
\begin{align}\label{DD-2-0-4}
\begin{aligned}
\int_0^{\infty} \mathcal{O}_{ij} dt & = \int_B \mathcal{O}_{ij} dt
+ \int_{\bbr_+ \backslash B} \mathcal{O}_{ij} d = \int_B \mathcal{O}_{ij} dt
+ \int_{\bbr_+ \backslash B} \mathcal{M}_{ij} dt \quad ( \because \text{Definition of  } B)
\\ & \leq \int_B \frac{1}{N} \sum_{k=1}^N \left[\left(R_{ik}+R_{jk}-\max_{i,j}\frac{4|\tilde\lambda^0|}{\kappa_{ij}^0}\right)\left(\max_{i,j} \kappa_{ij}^0 + \frac{4\mu}{\gamma}
\right)\right] dt
+ \int_{\bbr_+ \backslash B} \mathcal{M}_{ij} dt  \quad (\because \eqref{D-2-0-0})
\\ & \leq \int_B \frac{1}{N} \sum_{k=1}^N \left[\left(2-\max_{i,j}\frac{4|\tilde\lambda^0|}{\kappa_{ij}^0}\right)\left(\max_{i,j} \kappa_{ij}^0 + \frac{4\mu}{\gamma}
\right)\right] dt
+ \int_{\bbr_+ \backslash B} \mathcal{M}_{ij} dt
\\ & = 2m(B)\left[\left(1-\max_{i,j}\frac{2|\tilde\lambda^0|}{\kappa_{ij}^0}\right)\left(\max_{i,j} \kappa_{ij}^0 + \frac{4\mu}{\gamma}
\right)\right]
+ \int_{\bbr_+ \backslash B} \mathcal{M}_{ij} dt
\\ & = 2m(B)\left[\left(1-\max_{i,j}\frac{2|\tilde\lambda^0|}{\kappa_{ij}^0}\right)\left(\max_{i,j} \kappa_{ij}^0 + \frac{4\mu}{\gamma}
\right)\right]
+ \int_0^\infty \mathcal{M}_{ij} dt<\infty.
\end{aligned}
\end{align}
Since each summand of $\mathcal{O}_{ij}$ is positive, \eqref{DD-2-0-4} implies 
\begin{align}\label{Z-2}
\int_0^{\infty} \left(\kappa_{ik}\left(R_{ik}-\max_{i,j}\frac{2|\tilde\lambda^0|}{\kappa_{ij}^0}\right)+\kappa_{jk}\left(R_{jk}-\max_{i,j}\frac{2|\tilde\lambda^0|}{\kappa_{ij}^0}\right)\right) dt < \infty, \quad \forall i,j,k \in \mathcal{N}.
\end{align}
Now, we use the relationship
\[
1-R_{ij} \leq \mathcal{L}_{ij} \leq \mathcal{L}_{ij}^0
\]
and a priori condition $\eqref{ZZ-1}$ to get
\[
R_{IJ} \geq 1-\max_{k,l}\mathcal{L}_{kl}^0 > \max_{i,j}\frac{2|\tilde\lambda^0|}{\kappa_{ij}^0},
\]
for any indices $I$ and $J$. Therefore by \eqref{Z-2} we have
\begin{align*}
\left(1-\max_{K,L}\mathcal{L}_{KL}^0 -\max_{I,J}\frac{2|\tilde\lambda^0|}{\kappa_{IJ}^0}\right)\int_0^{\infty} \left(\kappa_{ik}+\kappa_{jk}\right) dt < \infty, \quad \forall i,j,k \in \mathcal{N}.
\end{align*}
Again, by lemma \ref{L4.4}, we obtain
\begin{align}\label{DD-2-0-5}
\lim_{t \to \infty} \kappa_{ik}
= \lim_{t \to \infty} \kappa_{jk}=0.
\end{align}
On the other hand, we recall the result of \ref{E-2-2}:
\begin{align*}
\dot{R}_{ij}=&\frac{1}{2}(\dot{h}_{ij}+\dot{h}_{ji})\\
=& \frac{1}{N}\sum_{k=1}^N \Big[  \kappa_{jk}(R_{ik}-R_{jk}R_{ij}  )+\kappa_{ik}(R_{jk}-R_{ik}R_{ji}  )\big)
-2 \tilde{\lambda}_{jk} I_{jk}I_{ij} -2 \tilde{\lambda}_{ik}I_{ki}I_{ij}) \Big].
\end{align*}
From uniform boundedness of $\kappa_{ij}$,$\tilde{\lambda}_{ij}$,$R_{ij}$ and $I_{ij}$, we have uniform boundness of $\dot{R}_{ij}$. Combining this together with uniform boundedness of $\dot{\kappa}_{ij}$, we obtain uniform boundedness of
\[
 \ddot{\kappa}_{ij} =  -\gamma \dot{\kappa}_{ij} - 2 \mu \dot{R}_{ij},
\]
which leads to uniform continuity of $\dot{\kappa}_{ij}$. As integration of $\dot{\kappa}_{ij}$ is finite from \eqref{DD-2-0-5}:
\[
\int_{0}^{\infty} \dot{\kappa}_{ij}(s)ds = -\kappa_{ij}^0,
\]
again from by lemma \ref{L4.4}, we conclude
\begin{align}\label{D-2-0-6}
\lim_{t \to \infty}\dot{\kappa}_{ik}(t) = \lim_{t \to \infty}\dot{\kappa}_{jk}(t) = 0.
\end{align}
Therefore, we take limit $t \to \infty$ to the dynamics
\[
\dot{\kappa}_{ij} = -\gamma \kappa_{ij} + \mu \| z_i - z_j \|^2
\]
to find the desired result
\[
\lim_{t \to \infty}\|z_i(t) - z_j(t) \| = 0.
\]
\end{proof}
\vspace{0.2cm}

\subsection{Hebbian coupling law} \label{secL5.6:5.2}
In this subsection, we study emergent dynamics of \eqref{E-1-7} with Hebbian coupling law: 
\begin{align}\label{NE-5}
\begin{cases}
\displaystyle \dot{z}_j= \displaystyle\frac{1}{N}\sum_{k=1}^N\kappa_{jk}( z_k-R_{jk} z_j)+\frac{1}{N}\sum_{k=1}^N\tilde{\lambda}_{jk}(\langle z_j, z_k\rangle-\langle z_k, z_j\rangle)z_j,~~t > 0, \vspace{0.2cm} \\
\displaystyle \dot{\kappa}_{jk}=-\gamma\kappa_{jk}+\mu \left(\displaystyle 1-\frac{\|z_j -z_k \|^2}{2} \right),\quad \dot{\tilde{\lambda}}_{jk}=-\gamma\tilde{\lambda}_{jk}+\mu\tilde{\Gamma}(z_j, z_k), \vspace{0.2cm} \\
\displaystyle (z_j, \kappa_{jk}, \tilde{\lambda}_{jk})(0)= (z_j^0, \kappa_{jk}^0, \tilde{\lambda}_{jk}^0) \in \bbh\bbs^d \times \bbr_+ \times \bbr, \quad  j,k\in \mathcal{N}.
\end{cases}
\end{align}
Basically, we follow the same arguments in Section \ref{sec:4.2} to derive the emergent dynamics of system \eqref{NE-5}. \newline

\begin{proposition} \label{P5.1}
Suppose that the following relations hold
\begin{align*}
\tilde{\lambda}_{ij}^0 =\tilde\lambda^0,  \quad \tilde{\Gamma}(t)=0, \quad t>0, \quad i, j \in {\mathcal N},
\end{align*}
and that there exists a function $\tilde\lambda$, positive constants $\kappa_m, \kappa_M$ such that 
\begin{align}
\begin{aligned} \label{E-7}
& \tilde{\lambda}_{ij}(t) = \tilde\lambda(t),  \quad  i,j \in \mathcal{N}, \quad  t>0, \\
& \frac{1}{2}\kappa_M + 2|\tilde\lambda^0| < \kappa_m \leq \frac{2\mu}{\gamma}\cdot\frac{\kappa_M-\kappa_m +2|\tilde\lambda^0|}{\kappa_M},
\end{aligned}
\end{align} 
and let $(Z, K, {\tilde \Lambda})$ be a solution to \eqref{E-2} with initial data satisfying 
\begin{equation} \label{E-8}
\min_{i,j}\kappa_{ij}^0  > \kappa_m, \quad {\mathcal D}^0 < \frac{2 \kappa_m -4|\tilde\lambda^0|}{\kappa_M} -1,
\end{equation}
and a priori condition
\begin{align}\label{E-8-0}
{{\max_{i,j}\sup_{0 \leq t < \infty} \kappa_{ij}(t) \leq \kappa_M}}.
\end{align}
Then, there exist positive constants $D_2$ and $D_3$ such that 
\[
{\mathcal D}(Z(t)) \leq D_2 e^{-D_3 t}, \quad t>0.
\]
\end{proposition}

\begin{proof} ~As in a proof of Proposition \ref{P4.1}, by $\eqref{E-8}_1$ the set
\[
\bar{\mathcal{T}} := \{ \tau \in (0, \infty) ~ : ~ \min_{i,j}\kappa_{ij}(t) > \kappa_m, ~ \forall t \in (0, \tau) \}
\]
is nonempty. So $\bar{t}^* := \sup \bar{\mathcal{T}}$ is well defined in $(0, \infty ]$. We use the same argument in a proof of Lemma \ref{L5.2} to find
\begin{align*}
\dot{h}_{ij}+\dot{h}_{ji} 
= \frac{2}{N}\sum_{k=1}^N\Big( \kappa_{jk}(R_{ik}-R_{jk}R_{ij}  )+\kappa_{ik}(R_{jk}-R_{ik}R_{ji}  )\big)
-2 \tilde{\lambda}_{jk} I_{jk}I_{ij} -2 \tilde{\lambda}_{ik}I_{ki}I_{ij}) \Big).
\end{align*}
Then, we use the above relation to find
\begin{align*}
\begin{aligned}
\dot{{\mathcal D}}_{ij}(t) &=-\frac{1}{2}(\dot{h}_{ij}+\dot{h}_{ji})\\
&= -\frac{1}{N} \sum_{k=1}^N \Big( (\kappa_{ik}+\kappa_{jk}) {\mathcal D}_{ij}  \Big)
+ \frac{1}{N} \sum_{k=1}^N \Big( \kappa_{ik} {\mathcal D}_{ik} + \kappa_{jk} {\mathcal D}_{jk} \Big) {\mathcal D}_{ij} \\
&+ \frac{1}{N} \sum_{k=1}^N (\kappa_{ik} - \kappa_{jk} )( {\mathcal D}_{jk}- {\mathcal D}_{ik}) + \frac{1}{N} \sum_{k=1}^N 2(\tilde{\lambda}_{jk} I_{jk}I_{ij} + \tilde{\lambda}_{ik}I_{ki}I_{ij}) \\
& =: {\mathcal J}_{21} +  {\mathcal J}_{22} + {\mathcal J}_{23} + {\mathcal J}_{24}.
\end{aligned}
\end{align*}
Below, we estimate the term ${\mathcal J}_{2i}$ separately. \newline

\noindent $\bullet$~(Estimate of ${\mathcal J}_{2i},~i= 1,2,3$):~In this case, all the terms are exactly the same as in \eqref{D-2-1}, we can use the same argument as in \eqref{D-2-1} to find
\begin{equation}\label{NE-5-1}
{ \mathcal{J}_{21} \geq 2\kappa_m {\mathcal D} }, \quad  \mathcal{J}_{22} \leq 2\kappa_M {\mathcal D}^2, \quad 
 \mathcal{J}_{23} \leq 2(\kappa_M - \kappa_m) {\mathcal D},
 \quad  t \in [0, \bar{t}^*).
\end{equation}

\vspace{0.2cm}

\noindent $\bullet$~(Estimate of ${\mathcal J}_{24}$):~We use the estimate
\[
I_{ij} \leq | I_{ij} | = \sqrt{1-R_{ij}^2} = \sqrt{(1-R_{ij})(1+R_{ij})} \leq \sqrt{2(1-R_{ij})} = \sqrt{2{\mathcal D}_{ij}} \leq \sqrt{2{\mathcal D}}
\]
to find
\begin{equation}\label{NE-6}
\mathcal{J}_{24} \leq 8 {\mathcal D} | \tilde \lambda | = 8 {\mathcal D}| e^{-\gamma t}\tilde \lambda^0 | < 8 {\mathcal D} |\tilde \lambda^0|.
\end{equation}
As in a proof of Proposition \ref{P4.1}, for each $t$, we assume that indices $i,j$ are chosen to satisfy
\[  {\mathcal D} = {\mathcal D}_{ij}. \]
We now combine \eqref{NE-5-1} and \eqref{NE-6} to get the Riccati type differential inequality:
\begin{equation} \label{NE-7}
\dot{\mathcal{D}} < -2(2\kappa_m-\kappa_M-4|\tilde\lambda^0|) {\mathcal D}+2\kappa_M {\mathcal D}^2, \quad t \in [0, \bar{t}^* ),
\end{equation}
where we use the first inequality of \eqref{E-7} to see
\[  2\kappa_m-\kappa_M-4|\tilde\lambda^0|>0. \]
We apply the comparison principle to \eqref{NE-7} to find
\begin{equation} \label{NE-8}
{\mathcal D}(t) \leq \frac{1}{\left(\frac{1}{{\mathcal D}^0}-\frac{2\kappa_m - \kappa_M -4|\tilde\lambda^0|}{\kappa_M}\right)e^{(2\kappa_m-\kappa_M -4|\tilde\lambda^0| )t}
+\frac{2\kappa_m-\kappa_M -4|\tilde\lambda^0|}{\kappa_M}}, \quad t \in [0, \bar{t}^*).
\end{equation}
Next, we will verify
\[ \bar{t}^*=\infty. \]
Since we will use proof by contradiction, suppose that $\bar{t}^*$ is finite. We use the initial condition 
\[ {\mathcal D}^0 < \frac{2 \kappa_m -4|\tilde \lambda^0|}{\kappa_M} -1, \]
to see that ${\mathcal D}$ is decreasing in $ t \in [0, {\bar t}^*)$ from \eqref{NE-7}. Therefore, one has
\begin{align*}
\begin{aligned}
\dot\kappa_{ij} &= -\gamma\kappa_{ij} + \mu - \mu {\mathcal D}_{ij} \geq -\gamma\kappa_{ij} + \mu - \mu {\mathcal D} \\
&\geq -\gamma \kappa_{ij} + 2\mu \frac{\kappa_M -\kappa_m +2|\tilde\lambda^0|}{\kappa_M}, \quad t \in [0, {\bar t}^*).
\end{aligned}
\end{align*}
By comparison principle, one obtains
\begin{align*}
\begin{aligned}
\kappa_{ij} &\geq \left( \kappa_{ij}^0 - \frac{2\mu}{\gamma}\frac{\kappa_M-\kappa_m +2|\tilde\lambda^0|}{\kappa_M}\right)e^{-\gamma t} + \frac{2\mu}{\gamma}\left(\frac{\kappa_M-\kappa_m+2|\tilde\lambda^0|}{\kappa_M}\right) \\
&\geq \left( \kappa_{ij}^0 - \kappa_m\right)e^{-\gamma t} + \kappa_m, \quad t \in [0, \bar{t}^*),
\end{aligned}
 \end{align*}
where the last equality holds from the second inequality of \eqref{E-7}. From definition of $\bar{\mathcal{T}}$, there exist indices $k$ and $l$ such that 
\[ \kappa_m = \lim_{t \nearrow \bar{t}^*} \kappa_{kl}. \]
Since $\bar{t}^*$ is finite, we have following inequality:
\begin{align*}
\kappa_m = \lim_{t \nearrow {\bar t}^*} \kappa_{kl}
\geq \lim_{t \nearrow {\bar t}^*} \left( \kappa_{kl}^0 - \kappa_m\right)e^{-\gamma t} + \kappa_m
= \left( \kappa_{kl}^0 - \kappa_m\right)e^{-\gamma \bar{t}^*} + \kappa_m
> \kappa_m,
 \end{align*}
 which is a contradictory. Thus $\bar{t}^* = \infty$ and 
 \[  \min_{i,j}\kappa_{ij}(t) > \kappa_m, \quad t \in [0, \infty). \]
Then the relation \eqref{NE-8} implies our desired estimate.
\end{proof}

\noindent \begin{proof}[{\bf Proof of Theorem \ref{T3.4}}]Recall the conditions \eqref{Y-0}: 
\begin{equation} \label{Y-1}
2|\tilde{\lambda}^0| < \kappa < \min \left\{  \frac{\mu}{\gamma},~\min_{i,j}\kappa_{ij}^0 \right \}, \quad 
\max \left\{  \max_{i,j} \kappa^0_{ij}, ~\frac{\mu}{\gamma} \right \} \leq \frac{2\mu (\kappa-2|\tilde{\lambda^0}|)}{2\mu-\gamma \kappa}, 
\quad {\mathcal D}(Z^0) < 1-\frac{\gamma}{\mu} \kappa.
\end{equation}
Now, it suffices to show that the above conditions satisfy \eqref{E-8} and \eqref{E-8-0}, i.e.,
\begin{align}\label{Y-2}
\kappa_m < \min_{i,j} \kappa^0_{ij}, \quad {\mathcal D}(Z^0) < \frac{2 \kappa_m-4|\tilde{\lambda}|}{\kappa_M} -1, \quad \sup_{0 \leq t < \infty} \max_{i,j} \kappa_{ij}(t) \leq \kappa_M.
\end{align}
We first figure out $\kappa_m$ and $\kappa_M$ satisfying \eqref{E-7}:
\begin{equation}\label{Y-3}
\frac{1}{2}\kappa_M + 2|\tilde\lambda^0| < \kappa_m, \quad \kappa_m \leq \frac{2\mu}{\gamma}\cdot\frac{\kappa_M-\kappa_m +2|\tilde\lambda^0|}{\kappa_M}.
\end{equation}
Since $\kappa$ is a candidate of $\kappa_m$, we will assume that $\kappa_m$ satisfies \eqref{Y-1}. Rewriting $\eqref{Y-3}_2$, we have 
\begin{align}\label{Y-4}
 \kappa_m \leq \frac{2\mu}{\gamma}\cdot\frac{\kappa_M-\kappa_m +2|\tilde\lambda^0|}{\kappa_M}
 \Longleftrightarrow
  \frac{2\mu (\kappa_m-2|\tilde\lambda^0|)}{2\mu-\gamma \kappa_m} \leq \kappa_M.
\end{align}
Optimizing $\kappa_M$ under \eqref{Y-4}, we have
\[
\frac{2\mu (\kappa_m-2|\tilde\lambda^0|)}{2\mu-\gamma \kappa_m} = \kappa_M.
\]
Therefore, as we set
\begin{equation} \label{Y-5}
\kappa_m= \kappa \quad \mbox{and} \quad \kappa_M= \frac{2\mu (\kappa-2|\tilde\lambda^0|)}{2\mu-\gamma \kappa},
\end{equation}
$\eqref{Y-3}_2$ is achieved. In particular, as $\kappa$ satisfies $\eqref{Y-1}_1$, we have
\[
\frac{1}{2}\kappa_M + 2|\tilde\lambda^0|< \kappa_m
\Longleftrightarrow
\kappa < \frac{\mu}{\gamma},
\]
which is true from $\eqref{Y-1}_1$. Hence $\eqref{Y-3}_1$ is achieved. \vspace{0.2cm}\\

\noindent $\bullet$~(Verification of \eqref{Y-2}):~Clearly, $\eqref{Y-1}_1$ implies $\eqref{Y-2}_1$. By the setting \eqref{Y-5}, one has 
\[  \frac{2 \kappa_m-4|\tilde\lambda^0|}{\kappa_M} -1 = \frac{2\kappa-4|\tilde\lambda^0|}{\frac{2\mu (\kappa-2|\tilde\lambda^0|)}{2\mu-\gamma \kappa}} - 1 = 1- \frac{\gamma \kappa}{\mu}. \]
Hence $\eqref{Y-1}_3$ is equivalent to $\eqref{Y-2}_2$ under the setting \eqref{Y-5}. As we did in a proof of Theorem \ref{T3.2}, $\eqref{Y-2}_3$ is verified by the estimate 
\begin{align*}
\begin{aligned}
\kappa_{ij}(t) &= e^{-\gamma t}\left(\kappa_{ij}^0 + \int_0^t \mu e^{\gamma s}R_{ij} ds \right) 
\leq  e^{-\gamma t}\left(\kappa^0_{ij} + \int_0^t \mu e^{\gamma s} ds \right)  \\
&\leq \left(\kappa^0_{ij} -\frac{\mu}{\gamma} \right)e^{-\gamma t} + \frac{\mu}{\gamma} \leq \max \left\{ \max_{i,j} \kappa^0_{ij}, ~\frac{\mu}{\gamma} \right\} \leq \kappa_M.
\end{aligned}
\end{align*}
Finally, we can apply the result of Proposition \ref{P5.1} to derive the desired estimate.\end{proof}

\section{Conclusion} \label{sec:6}
\setcounter{equation}{0}
In this paper, we have studied the emergent dynamics of the LHS model with adaptive coupling gains. When the dynamics of coupling gains are decoupled from the dynamics of state, say, they are simply constants, in previous literature, several sufficient frameworks were proposed for complete aggregation in which all states collapse to the same state. However, when coupling gains and state evolutions are intertwined via adaptive coupling laws, emergent dynamics are more delicate and interesting. In order to couple the dynamics of coupling gains and state, we employ two types of coupling laws, namely anti-Hebbian law and Hebbian law in analogy with the dynamics of brain neurons. The former causes the increment of coupling gain, as the state differences become larger, whereas the latter causes the opposite effect. In the case of the same free flow for all particles, states aggregate to the same state asymptotically for some class of initial data and system parameters. When rotational coupling gain is the minus of the half of the sphere coupling gain, our first result says that the relative state tends to zero and coupling gains tend to zero asymptotically. Since the coupling gain becomes smaller over time, analysis of complete aggregation is highly nontrivial and difficult to analyze. 
Despite this apparent difficulty, we use the Lyapunov functional approach and Barbalat's lemma to show that the relative states and coupling gains tend to zero for the anti-Hebbian case. For the Hebbian coupling case, we show that the square of the state diameter tends to zero exponentially fast for some admissible class of initial data and initial system parameters. The same things can be done for an asymptotically SL coupling gain pair. All presented results in this paper deal with the ensemble of particles with the same free flows. There are several issues that have not been addressed in this work. For example, for an ensemble of LHS particles with the same free flow, a bi-polar state can emerge as one of the resulting asymptotic patterns. Then, is this bi-polar configuration unstable as for the Lohe sphere model?  For the ensemble of particles with heterogeneous free flows, emergent dynamics is completely unknown even for nonnegative coupling gains, not to mention adaptive coupling gains. These interesting issues will be left for future work. 




\end{document}